\documentclass[11pt]{article}

\usepackage{amsmath,amsthm,amsfonts,amssymb,mathtools}
\usepackage{float}
\usepackage[svgnames,dvipsnames]{xcolor}
\usepackage[shortlabels]{enumitem}
\usepackage{graphicx}
\usepackage{epstopdf}
\usepackage{tikz-cd}
\usepackage{tcolorbox}
\usepackage{amsbsy,bm}
\usepackage{mathrsfs}

\definecolor{dullmagenta}{RGB}{102,0,102}
\def\col{dullmagenta} 

\usepackage[hidelinks,
hypertexnames=false,
    colorlinks=false,       
    linkcolor=\col,          
    citecolor=\col,      
    filecolor=\col,       
    urlcolor=\col,            
    ]{hyperref}

\makeatletter 
\renewcommand\paragraph{\@startsection{paragraph}{4}{\z@}
                                      {\parskip}
                                      {-1em}%
                                      {\normalfont\normalsize\bfseries}}
\makeatother



\parskip=\medskipamount
\newcommand{\lag}{\mathfrak{g}}

\def\lagk{\lag^k}
\def\lagdk{(\lag^*)^k}

\newcommand{\R}{\mathbb{R}}

\newcommand{\disV}{\mathcal{V}}
\newcommand{\disW}{\mathcal{W}}
\newcommand{\disR}{\mathcal{R}}

\def\fpd#1#2{{\displaystyle\frac{\partial #1}{\partial #2}}}

\newcommand{\st}{\;\ifnum\currentgrouptype=16 \middle\fi|\;}

\makeatletter
\def\smallunderbrace#1{\mathop{\vtop{\m@th\ialign{##\crcr
   $\hfil\displaystyle{#1}\hfil$\crcr
   \noalign{\kern3\p@\nointerlineskip}%
   \tiny\upbracefill\crcr\noalign{\kern3\p@}}}}\limits}

\makeatletter 
\def\@endtheorem{\endtrivlist}
\makeatother

\makeatletter
\def\th@plain{%
  \thm@notefont{}
  \itshape 
}
\def\th@definition{%
  \thm@notefont{}
  \normalfont 
}
\makeatother

\theoremstyle{plain}
\newtheorem{theorem}{Theorem}
\newtheorem{lemma}{Lemma}
\newtheorem{proposition}{Proposition}
\newtheorem{corollary}{Corollary}

\theoremstyle{definition}
\newtheorem{remark}{Remark}
\newtheorem{example}{Example}

\newtheorem{definition}{Definition}

\begin{document}

\date{}
\title{Conditions  for symmetry reduction of polysymplectic and polycosymplectic structures}

\author{E.\ Garc\'{\i}a-Tora\~{n}o Andr\'{e}s\textsuperscript{a,b} and T.\ Mestdag\textsuperscript{c}\\[2mm]
{\small \textsuperscript{a} Departamento de Matem\'atica,
Universidad Nacional del Sur (UNS),}  \\
{\small  Av.\ Alem 1253, 8000 Bah\'ia Blanca, Argentina}\\[1mm]
{\small \textsuperscript{b} Instituto de Matem\'atica de Bah\'ia  UNS-CONICET, }\\
{\small  Universidad Nacional del Sur,}\\
{\small  Av.\ Alem 1253, 8000 Bah\'ia Blanca, Argentina}\\[1mm]
{\small \textsuperscript{c} Department of Mathematics,  University of Antwerp,}\\
{\small Middelheimlaan 1, 2020 Antwerpen, Belgium}\\
}

\maketitle

\begin{abstract}
For Hamiltonian field theories on polysymplectic manifolds  with a symmetry group action and a momentum map, we explore the redundancy in a set of necessary conditions that has appeared in the literature, for a generalized version of the Marsden-Weinstein symmetry reduction theorem. Next, we prove a necessary and sufficient condition for polycosymplectic reduction. We relate polycosymplectic reduction in a one-to-one way to  the reduction of an associated larger polysymplectic manifold. Throughout the paper, we provide examples and discuss special cases.

\vspace{3mm}


 \vspace{3mm}


\end{abstract}

\section{Introduction}\label{sec:Introduction}

The last decades have seen an ever growing interest in classical field theories with a symmetry group, and in the application of symmetry reduction techniques. Much of this interest is motivated by the wide range of applications that these methods have found in mathematical physics in general, and in geometric mechanics in particular. 

Since classical field theories find their origin in the calculus of variations, most of the modern literature on their reduction is typically set in a Lagrangian framework and is aimed at extending so-called ``Lagrange-Poincar\'e reduction''  from mechanics \cite{LagRedbyStag} to field theories  \cite{Castrillon2,LagPoincare_JGP,LTM_LP}. In a nutshell, one uses the symmetry to reduce the configuration space to a quotient space, and  to define a variational principle for a reduced Lagrangian field theory there. Remarkably, the equivalence between the solutions of the reduced and the original field theory is not as neat as it is in mechanics, and obstructions to the reconstruction of solutions occur.

It is well-known that symplectic structures lie at the basis of the Hamiltonian formulation of classical mechanics. Besides variational methods, there exist also geometric approaches to field theories that rely on generalizations of symplectic geometry. One of them, is the so-called polysymplectic formalism. Introduced by G{\"{u}}nther~\cite{Gunter}, it makes use of a family of $k$ closed 2-forms $\omega^a$ on a manifold $M$ which share a generalized nondegeneracy condition (see Section~\ref{sec:polysympletic} for the preliminaries). The corresponding Hamiltonian field theory of $(M,\omega^a)$ follows from considering the integral sections of a family of vector fields that satisfies a symplectic-type Hamiltonian equation. Locally, these equations correspond to the Hamilton-De~Donder-Weyl equations: If $\psi^i(t)$ represents the sought-for field, depending on $k$ parameters $t^a$, and if $H(q^i,p_i^a)$ is the Hamiltonian, we want to find the  first part of a solution $(\psi^i(t), \psi_i^a(t))$ of
\begin{equation}\label{eq:HdDW}
\left.\fpd{H}{q^i}\right|_{\varphi(t)}=-\left.\sum_{a=1}^k\fpd{\psi^a_i}{t^a}\right|_{t},\qquad \left.\fpd{H}{p^a_i}\right|_{\varphi(t)}=\left.\fpd{\psi^i}{t^a}\right|_{t}.
\end{equation}

An important point to make for this paper is that polysymplectic field theories are restricted to Hamiltonians $H(q^i,p_i^a)$ that do not depend explicitly on the independent parameters (for physical field theories, these are typically the space-time coordinates). One may include Hamiltonians of the type $H(t^a,q^i,p_i^a)$   by considering $k$-cosymplectic geomety \cite{Hamiltoniansystemskcosymplectic} or, more in general, polycosymplectic geometry. Essentially we need to bring $k$ extra closed one-forms $\eta^a$ in the picture (see Section~\ref{sec:reduction} for the precise definition), but the base line is the following: Just like a polysymplectic structure $(M,\omega^a)$ is a generalization of a symplectic structure, polycosymplectic geometry $(M,\omega^a,\eta^a)$ is inspired by cosymplectic geometry. 

Besides polysymplectic and polycosymplectic structures, there exist also other approaches to Hamiltonian field equations, such as e.g.\ multisymplectic geometry (see e.g. \cite{ReviewultisymplecticNarciso,InvitationMultisymplectic} for two review papers). The reason for our choice for the poly-approach is that it has the advantage that the field equations are formulated entirely in terms of forms and vector fields, for which the usual Cartan calculus can easily be applied. This observation  explains why many results from classical mechanics have rapidly found their analogues in the poly-formulations. The monograph~\cite{bookpoly}, with its comprehensive reference list, is maybe the best testament of this assertion.  Also in the current paper, we will make good use of the previous observation: we will investigate symmetry reduction from the perspective of the reduction of the relevant forms and vector fields (rather than of the associated variational principle).

Already G{\"{u}}nther  proposed a polysymplectic analog of the Marsden-Weinstein Theorem in symplectic geometry~\cite{SymplecticReduction}, although unfortunately the proof is known to be flawed. The first complete proof of a polysymplectic reduction theorem can be found in~\cite{Polyreduction}, expanding and correcting a previous attempt~\cite{MunSal}. Under modest assumptions on the action and the existence of a momentum map, \cite{Polyreduction} gives sufficient conditions for the reduced manifold to be polysymplectic (see the current Theorem~\ref{thm:polysymred2}). The authors  also show that these conditions are satisfied in a number of natural examples of polysymplectic manifolds. Later, Blacker~\cite{Blackerpoly} identified the main obstruction for the reduced space to be polysymplectic, and gave a necessary and sufficient condition for this to happen (see Theorem~\ref{thm:polysymred}); see also~\cite{NicoMartinez} for a closely related earlier result. As the first main goal of this paper, we will use Blacker's result to show (in Section~\ref{sec:polysympletic}) that one of the sufficient conditions of \cite{Polyreduction} is in fact redundant (Proposition~\ref{pro:A2enough}). We also indicate why we think that redundant condition made its appearance in \cite{Polyreduction} and we hint why there may be no need to include it in the statement of Theorem~\ref{thm:polysymred2}.

Until very recently, there did not exist a version of the Marsden-Weinstein reduction theorem for polycosymplectic structures in the literature (apart from a brief announcement in \cite{polycosymplecticmarreroreduction}). We state such a version in Theorem~\ref{thm:main}, by following closely Blacker's strategy of the polysymplectic case~\cite{Blackerpoly}. We first describe (in Proposition~\ref{pro:linearreduction}) how the reduction works at the linear level, and then discuss how this applies to the level sets of a momentum map in polycosymplectic geometry. As a result, we obtain the neccesary and sufficient conditions for the quotient to be polycosymplectic. Albert's cosymplectic reduction theorem~\cite{Albert} then follows as  Corollary~\ref{cor:albert} from Theorem~\ref{thm:main}. The reduction of the Hamiltonian dynamics is also briefly discussed in Theorem~\ref{thm:maindynamics}. 

In Section~\ref{sec:relation} we relate the two previous sections: we show that, if a polycosymplectic manifold is given, it is possible to define a polysymplectic structure on the space $\tilde M=M\times \R$ (which we will call below ``the lift'' of $M$, for convenience) as follows: 
\begin{center}
\begin{tikzcd}[row sep=large, column sep = 7ex]
(M,\omega^a,\eta^a) \arrow{r} & (\tilde M=M\times\R,\,\tilde \omega^a={\rm pr}^*\omega^a+ds\wedge {\rm pr}^*\eta^a).
\end{tikzcd} 
\end{center}
In view of the presence of obstructions for reduction in Theorem~\ref{thm:polysymred} and Theorem~\ref{thm:main}, we then show in Proposition~\ref{pro:equivalencepolysympolyco} that, with the above construction, its lift $(M,\omega^a,\eta^a)$ can be reduced  if, and only if, $(\tilde M,\tilde \omega^a)$ can be reduced (with a suitable definition of a lifted action and momentum map). This has the immediate advantage that we can easily obtain a sufficient condition for polycosymplectic reduction of $(M,\omega^a,\eta^a)$, based on the one (now the only remaining one, in view of Proposition~\ref{pro:A2enough}) for polysymplectic reduction of $(\tilde M,\tilde \omega^a)$  (see Theorem~\ref{thm:sufficient}). The section ends with a discussion about the relation between the Hamiltonian dynamics on the two spaces $M$ and $\tilde M$. In the last section we discuss as examples the peculiarities of three special cases: those of 
$k$-symplectic and $k$-cosymplectic manifolds, of the stable cotangent bundle and of the product of cosymplectic manifolds.

We conclude this Introduction with a few remarks on the recent paper \cite{Poly-deLucas} (an arXiv preprint upon submission of this work) on reduction of polycosymplectic structures, which is closely related to ours in its aims and results. We would like to point out a few differences, though. In \cite{Poly-deLucas} there is no direct proof of the reduction theorem, but a proof that is based on an extension of the polycosymplectic structure to a new polysymplectic structure on the manifold $\hat M=M\times\R^k$, as follows
\begin{center}
\begin{tikzcd}[row sep=large, column sep = 7ex]
(M,\omega^a,\eta^a) \arrow{r} & (\hat M=M\times\R^k,\,\hat \omega^a={\rm pr}^*\omega^a+ds^a\wedge {\rm pr}^*\eta^a),
\end{tikzcd} 
\end{center}
where $s^a$ is the coordinate in the $a$-th factor of $\R^k$. It is then described in \cite{Poly-deLucas} how the sufficient conditions of \cite{Polyreduction}  for the reduction of this new space $\hat M$ can be translated to sufficient conditions for the reduction of the original polycosymplectic manifold $M$. The relation with the necessary and sufficient conditions of  \cite{Blackerpoly} is not explored in  \cite{Poly-deLucas}.

It is obvious that $\hat M$ is a different and larger polysymplectic manifold than our lift $\tilde M$: both constructions yield the same cosymplectic manifold in the case $k=1$, but differ when $k\geq 2$.  
There are, in our opinion, two desirable and natural properties of $\tilde M$ that $\hat M$ does not have. First, our lift of the stable cotangent bundle of $k^1$-covelocities, which plays the role of the canonical phase space in the polycosymplectic Hamiltonian field theory, has as its lift the cotangent bundle of $k^1$-covelocities, which is the canonical phase space in the polysymplectic Hamiltonian field theory  (see Section \ref{sec:kcosymplectic}). Second,  we show in Proposition~\ref{pro:liftstoksympl} that in the case where $M$ is a $k$-cosymplectic manifold  our lift to $\tilde M$ remains a $k$-symplectic manifold. Remark that we have a direct proof for polycosymplectic Marsden-Weinstein reduction and that we have used our lift only to show  the equivalence between the necessary and sufficient conditions for polycosymplectic reduction of $M$ and for polysymplectic reduction of $\tilde M$. As a result, we can recover exactly the same  sufficient conditions for reduction that appear in Theorem~6.13 of \cite{Poly-deLucas}, by making use of the larger lift to $\hat M$ (even though the statement in \cite{Poly-deLucas} uses a more general definition of momentum map). But, as a consequence of our Proposition~\ref{pro:A2enough}, we believe one of them to be redundant. 

\paragraph{Notations and conventions.} We will use Einstein's convention on the indices $i,j,k,\dots$, but \emph{not} on the indices $a,b,c,\dots$  (sum will be written, if required). Indices $a,b,c,\dots$ will run from $1$ to $k$ unless otherwise specified. All Lie groups in this paper are assumed to be connected. For concreteness, actions will be assumed to be on the left.

\section{A sufficient condition for polysymplectic reduction}\label{sec:polysympletic}

There are a few similar, but not entirely equivalent, definitions of a polysymplectic manifold. We will use the following, which essentially coincides with the one introduced in~\cite{Gunter}:

\begin{definition}\label{def:polysymplectic}  A \emph{$k$-polysymplectic structure}, or simply a \emph{polysymplectic structure}, on a manifold $M$ is a family of closed two-forms $\omega^1,\dots,\omega^k$, such that
\[
\ker\omega^1\cap\dots\cap\ker \omega^k=\{0\}. 
\]
\end{definition}

Definition~\ref{def:polysymplectic} is also used in \cite{Polyreduction,Blackerpoly}, which will be our main sources  in this section. Following G{\"{u}}nther, we say that the polysymplectic structure is standard if 
around each point of $M$ there exists a coordinate chart $(q^i,p^a_i)$ such that
\begin{equation}\label{eq:darboux-standard}
\omega^a=dq^i\wedge dp_i^a,\qquad  (a=1,\ldots,k). 
\end{equation}
In analogy with the symplectic case, coordinates as in~\eqref{eq:darboux-standard} are called Darboux coordinates, but for a general polysymplectic manifold they may not exist. Equivalently, a polysymplectic manifold $(M,\omega^a)$ is standard if it is locally isomorphic (in the polysymplectic sense) to the canonical model 
\[
(T^1_k)^*Q =T^*Q\oplus\underset{k\; {\rm copies }}{\dots}\oplus T^*Q,
\]
which is often referred to as the cotangent bundle of $k^1$-covelocities. In particular, for a standard polysymplectic manifold $(M,\omega^a)$ one has $\dim M=(k+1)n$ for some integer $n$. A $1$-polysymplectic manifold is a symplectic manifold (and it is always standard due to the Darboux theorem in symplectic geometry).

\begin{remark} To avoid possible confusion in terminology, we would like to point out that the term ``polysymplectic structure'' was also used in a completely different meanings and contexts, starting with
e.g. \cite{bookSardanashvily,bookLagHamMethods} and \cite{Kanatchikov1}. There it stands for a certain vector-valued form defined on an associated bundle of a given fibre bundle. In this paper we do not make use of this different description of classical field theories.
\end{remark}

We will assume that the reader is familiar with the notions of a $k$-vector field $\pmb{X}=(X_1,\dots,X_k)$ on a manifold $M$ and of an integral section $\phi\colon \R^k\to M$ of a given $k$-vector field through a point $x\in M$ (if not, see e.g.\ \cite{bookpoly}). We briefly recall that a $k$-vector field $\pmb{X}$ is a section of the vector bundle $T^1_kM\to M$, and that it is customary to denote its components (which are vector fields on $M$) by $X_a$, as above. It is also worth mentioning that, while the term integral section is well-established in this context, an integral section $\phi$ is not a section of any bundle in a natural way.

We also recall that a $k$-vector field is said to be integrable if it admits an integral section through each point of $M$. Given a  function $H\colon M\to \R$ on a polysymplectic manifold $(M,\omega^a)$, the $k$-symplectic Hamilton's equations for $H$ are 
\begin{equation}\label{eq:k-Sym}
\flat_{\omega} (\pmb{X})=dH, \tag{k-Sym}
\end{equation}
where $\flat_\omega(\pmb{X})$ is the 1-form on $M$ defined as:
\[
\flat_\omega(\pmb{X})=i_{X_1}\omega^1+\dots+ i_{X_k}\omega^k.
\]
We draw the attention of the reader that even though we convey to  the standard terminology of ``$k$-symplectic Hamilton's equations'', these are equations for a $k$-vector field on a \emph{general} polysymplectic manifold $(M,\omega^a)$, which need not be a $k$-symplectic manifold ($k$-symplectic manifolds are a particular case of polysymplectic manifolds; see Section~\ref{sec:kcosymplectic}).

Each solution $\pmb{X}$ of~\eqref{eq:k-Sym} is called a Hamiltonian $k$-vector field. Integral sections of such vector fields (if they exist) provide solutions to the  Hamilton-De~Donder-Weyl equations in field theories~\eqref{eq:HdDW} (see the Introduction). We point out that, with this terminology, a Hamiltonian $k$-vector field need not be integrable.

We will now recall the main results of~\cite{Polyreduction} and \cite{Blackerpoly}, concerning the reduction of polsysymplectic structures. 
An action $\Phi_g\colon M\to M$ of a Lie group $G$ on a polysymplectic manifold $(M,\omega^a)$ is a polysymplectic action if, for each $g\in G$, we have $\Phi_g^*\omega^a=\omega^a$. A momentum map for such an action is a map
\[
J=(J^1,\dots,J^k)\colon M\to\lagk\equiv \lag\times\underset{k\; {\rm copies }}{\dots}\times\lag 
\]
which satisfies
\begin{equation}\label{eq:momentumdefinition}
\left.i_{\xi_M}\omega^a\right|_x=dJ^a_\xi(x),  
\end{equation}
for each $a$ and each $\xi\in\lag$, where $J^a_\xi\colon M\to\R$ is the function $J^a_\xi(x)=\langle J^a(x),\xi\rangle$. Here $\xi_M\in\mathfrak{X}(M)$ is the infinitesimal generator of $\xi$, for the action $\Phi_g$. The momentum map is equivariant if  
\begin{equation}\label{eq:equivariance}
J(\Phi_g(x))=({\rm Ad}^*_{g^{-1}} J^1(x),\dots,{\rm Ad}^*_{g^{-1}} J^k(x)),
\end{equation}
where ${\rm Ad}_{g^{-1}}^*\colon \lag^*\to\lag^*$ is the coadjoint action. In other words, $J$ is a $G$-equivariant map when one endows $\lagdk$ with the so-called $k$-coadjoint action ${\rm Coad}^k_g\colon \lagdk\to \lagdk$:
\[
{\rm Coad}^k_g(\mu_1,\dots,\mu_k)= ({\rm Ad}^*_{g^{-1}} \mu_1,\dots,{\rm Ad}^*_{g^{-1}} \mu_k).
\]
The isotropy group of an element $\mu=(\mu_1,\dots,\mu_k)\in\lagdk$ under the $k$-coadjoint action is denoted $G_\mu$, and it is easy to see that
\[
G_\mu=G_{\mu_1}\cap\dots\cap G_{\mu_k}, 
\]
where each $G_{\mu_a}$ is the isotropy group of $\mu_a\in\lag^*$ under the usual ($k=1$) coadjoint action. The Lie algebra of $G_\mu$ will be denoted $\lag_\mu$. A polysymplectic manifold equipped with a polysymplectic action and an equivariant momentum map is called a polysymplectic Hamiltonian $G$-space.

Throughout the paper, we will assume that the action  $\Phi_g\colon M\to M$ is free and proper, even though this condition can often be relaxed. In that case $M\to M/G$ is a principal fiber bundle with vertical subbundle $\tilde \lag=\cup_{x\in M}\tilde\lag_x\subset TM$, where
\[
\tilde\lag_x= \{\xi_M(x)\st \xi\in\lag\}\subset T_xM.
\]
We will also use the notations
\[
 \tilde \lag_\mu=\bigcup_{x\in M}\left.\tilde\lag_\mu\right|_x\subset TM,\qquad \left.\tilde\lag_\mu\right|_x= \{\xi_M(x)\st \xi\in\lag_\mu\}\subset T_xM.
\]

If $S_x\subset T_xM$ is a subspace of the tangent space, the polysymplectic orthogonal of $S_x$ is the following subspace:
\begin{equation}\label{def:orthogonal}
S_x^\omega=\{v\in T_xM\st \omega^1(v,S_x)=\dots=\omega^k(v,S_x)=0\}\subset T_xM.   
\end{equation}
The main properties of the polysymplectic orthogonal can be found in~\cite{Blackerpoly}. A more general notion of $\ell$-th polysymplectic orthogonal (where only the first $\ell\leq k$ forms are considered in expression (\ref{def:orthogonal})) is studied in~\cite{Submanifoldsk}. We remark that, unlike in the symplectic case, only the inclusion $S_x\subset S_x^{\omega\omega}$ is guaranteed (we use the notation $S_x^{\omega\omega}= (S_x^{\omega})^{\omega}$).

The following result has been proved in~\cite{Blackerpoly}, Theorem~3.22 (But, beware of an unfortunate typo in the statement):

\begin{theorem}[Polysymplectic reduction theorem] \label{thm:polysymred} Let $(M,\omega^a,\Phi_g,J)$ be a polysymplectic Hamiltonian $G$-space, and $\mu\in\lagdk$ a regular value of $J$. Assume that $G_\mu$ acts freely and properly on $J^{-1}(\mu)$. Let  $\pi_\mu\colon J^{-1}(\mu)\to J^{-1}(\mu)/G_\mu$ be the canonical projection and $j_\mu\colon J^{-1}(\mu)\to M$  the canonical inclusion.

Then, the reduced space $M_\mu=J^{-1}(\mu)/G_\mu$ admits a unique polysymplectic structure $\omega^a_\mu$ satisfying $\pi_\mu^*\omega^a_\mu=j_\mu^*\omega^a$, 
if, and only if, for each $x\in J^{-1}(\mu)$, the following condition holds:
\begin{equation}\label{eq:nondegenerate}
\left.\tilde\lag_\mu\right|_x=\tilde\lag_x^\omega \cap \tilde\lag_x^{\omega\omega}.
\end{equation}
\end{theorem}

It will be useful to review the need for including the necessary and sufficient condition~\eqref{eq:nondegenerate}. From the definition of momentum map~\eqref{eq:momentumdefinition}, one can show that the tangent space of $J^{-1}(\mu)$ coincides with the polysymplectic orthogonal of $\tilde\lag$:
\[
T_x J^{-1}(\mu) =\{\xi_M(x)\st \xi\in\lag\}^\omega=\tilde\lag_x^\omega.
\]
Indeed,
\[
J^{-1}(\mu)=\{x\in M\st J^a_\xi(x)=\langle \mu^a,\xi\rangle, \text{ for all }\;\xi \in \lag \text{ and } a=1,\ldots,k\}. 
\]
Then, $v\in T_xJ^{-1}(\mu)$ if, and only if, $0=T_xJ^a_\xi(v)=\langle d J^a_\xi, v\rangle =\omega^a_x(\xi_M(x),v)=0$. So, $v\in \lag_x^\omega$. With that, one can show that
\begin{equation} \label{extranumber}
\left.\tilde\lag_\mu\right|_x= \tilde\lag_x \cap T_xJ^{-1}(\mu)  =\tilde\lag_x \cap (\tilde\lag_x)^\omega.
\end{equation}
Besides, it also follows that the family of closed two-forms $j_\mu^*\omega^a$ on $J^{-1}(\mu)$ satisfies
\[
\cap_{a=1}^k \big(\ker \left.j_\mu^*\omega^a\right|_x\big)=T_x J^{-1}(\mu)\cap(T_x J^{-1}(\mu))^\omega=\tilde\lag_x^\omega \cap \tilde\lag_x^{\omega\omega}.
\]

The condition for the family of two-forms $j_\mu^*\omega^a$ to define a polysymplectic structure when taking the quotient of $J^{-1}(\mu)$ by $G_\mu$ is $\left.\tilde\lag_\mu\right|_x = \cap_{a=1}^k \big(\ker \left.j_\mu^*\omega^a\right|_x\big)$. Therefore, the  neccesary and sufficient condition~\eqref{eq:nondegenerate} simply  removes this degeneracy on the quotient space $M_\mu$.

Since always $\tilde\lag_x\subset\tilde\lag_x^{\omega\omega}$, and in view of \eqref{extranumber}, the condition~\eqref{eq:nondegenerate}  can also be written as $\tilde\lag_x^\omega \cap \tilde\lag_x^{\omega\omega}\subset\left.\tilde\lag_\mu\right|_x$. The above result also appears in essence in \cite{Polyreduction}, after the proof of Lemma~3.6, although not written explicitly as a propositon there.
In the symplectic case, one has $\tilde\lag_x^{\omega\omega}=\tilde\lag_x$ and that is why there is no such condition in the well-known symplectic reduction theorem of~\cite{SymplecticReduction}. 

The dimension of the reduced polysymplectic space is
\[
\dim M_\mu=\dim M-k\cdot\dim G - \dim G_\mu.
\]
It follows that, in general, the reduced polysymplectic space $M_\mu$ will not be standard even if $M$ is. As a simple example, consider a Lie group $G$ and fix $\nu\in\lag^*$ and take $\mu=(\nu,\dots,\nu)\in\lagdk$ (in particular, $G_\mu= G_\nu$). The reduction of $M=(T^1_k)^*G$ w.r.t. the natural left action at $\mu$ has dimension $\dim M_\mu=\dim G-\dim G_\nu$ which is in general not of the form $(k+1)n'$ (with $n'$ a positive integer). This particular example is discussed in detail in \cite{Polyreduction}.

The paper \cite{Polyreduction} goes one step further than \cite{Blackerpoly}. One may find there a set of two sufficient conditions for reduction:
\begin{theorem}\label{thm:polysymred2} Let $(M,\omega^a,\Phi_g,J)$ be a polysymplectic Hamiltonian $G$-space, and $\mu\in\lagdk$ a regular value of $J$. Assume that $G_\mu$ acts freely and properly on $J^{-1}(\mu)$. If the following (sufficient) conditions hold: 
\begin{enumerate}
 \item[(A1)] $\ker (T_xJ^a)=T_x(J^{-1}(\mu))+\left.\ker\omega^a\right|_x+\left.\tilde\lag_{\mu_a}\right|_x$, for each $a=1,\dots,k$,
 \item[(A2)] $\left.\tilde\lag_\mu\right|_x=\cap_{a=1}^k \big(\left.\tilde\lag_{\mu_a}\right|_x+\left.\ker\omega^a\right|_x\big)\cap T_x(J^{-1}(\mu))$,
\end{enumerate}
then the reduced space $M_\mu=J^{-1}(\mu)/G_\mu$ admits a unique polysymplectic structure $\omega^a_\mu$ satisfying $\pi_\mu^*\omega^a_\mu=j_\mu^*\omega^a$. 
\end{theorem}

The above result says that condition~(\ref{eq:nondegenerate}) must follows from $(A1)$ and $(A2)$. As our first main result, we will show that the condition $(A1)$ is in fact redundant. We need to state a few lemmas about presymplectic vector spaces and manifolds before we can show this.

\begin{lemma}\label{lem:presymplectic} Let $(V,\Omega)$ be a presymplectic vector space and $S\subset V$ a subspace. Then: 
 \[
S^{\Omega\, \Omega}=S+ \ker\Omega. 
 \]
\end{lemma}

\begin{proof}  We construct the quotient space 
\[
\pi\colon V\to \overline{V}=V/\ker\Omega, 
\] 
whose elements are denoted $\overline{v}=\pi(v)$, $\overline{w}=\pi(w)$, etc. It is a symplectic vector space when furnished with the 2-form $\overline{\Omega}(\overline{v},\overline{w})=\Omega(v,w)$, naturally induced by $\Omega$. The symplectic orthogonal of a subspace $\overline{S}=\pi(S)\subset \overline{V}$ will be denoted $\overline{S}^{\overline\Omega}$. One easily verifies that the following property holds:
\begin{equation}\label{eq:quotientrelation2}
\overline{S^{\Omega}}=\overline{S}^{\,\overline\Omega}. 
\end{equation}
We now claim that 
\begin{equation}\label{eq:quotientrelation}
S^{\Omega}=\pi^{-1}\Big(\overline{S}^{\,\overline\Omega}\Big),  
\end{equation}
with $\overline{S}=\pi(S)$. Indeed, $v\in S^{\Omega}$ if $\Omega(v,S)=0$, and this implies $\overline{\Omega}(\overline{v},\overline{S})=0$; the other inclusion is similar. Recall that, for any subspace $\overline{S}\subset \overline{V}$, we have $\overline{S}^{\,\overline{\Omega}\,\overline{\Omega}}=\overline{S}$ because $\overline\Omega$ is symplectic. If we apply property~\eqref{eq:quotientrelation} to the subspace $S^{\Omega}$ we get 
\[
S^{\Omega \,\Omega}=\pi^{-1}\Big((\overline{S^{\Omega}})^{\,\overline\Omega}\Big) =\pi^{-1}\Big(\overline{S}^{\,\overline\Omega\,\overline\Omega}\Big)=\pi^{-1}(\overline{S})=S+\ker \Omega,
\]
as desired.
\end{proof}

The second lemma is a version of Proposition~4 of \cite{Reductionpresymplectic}, with the difference that we do not require the closed 2-form $\Omega$ to have constant rank. We include a proof for completeness.

\begin{lemma}\label{lem:presymplectickernel} Let $(M,\Omega)$ be a presymplectic manifold with presymplectic action and momentum map $J$. Let $j_\mu\colon J^{-1}(\mu)\to M$ be the natural inclusion. Then, we have
\[
\ker \left.j_\mu^*\Omega\right|_x=\ker \Omega_x+\left.\tilde\lag_{\mu}\right|_x,\quad x\in J^{-1}(\mu).
\]
\end{lemma}

\begin{proof} We first observe that \[
\ker \left.j_\mu^*\Omega\right|_x=\{v\in T_xJ^{-1}(\mu) \st \Omega_x(v,T_xJ^{-1}(\mu))=0\} = T_x(J^{-1}(\mu)) \cap (T_xJ^{-1}(\mu))^\Omega.   
\]
With the same reasoning as before, we get from the definition of the momentum map that
$
T_xJ^{-1}(\mu)=\tilde \lag_x^\Omega  
$
and therefore \[
\ker \left.j_\mu^*\Omega\right|_x= \tilde\lag_x^\Omega\cap \tilde\lag_x^{\Omega\,\Omega}. 
\]
If we now use Lemma~\ref{lem:presymplectic}, we see that 
\[
\ker \left.j_\mu^*\Omega\right|_x=  \tilde\lag_x^\Omega\cap (\lag_x+\ker\Omega_x)=(\tilde\lag_x^\Omega\cap \lag_x)+\ker\Omega_x =  \left.\tilde\lag_{\mu}\right|_x+\ker\Omega_x.
\]
The second equality follows from  $\ker\Omega_x\subset S_x^\Omega$ for any subspace $S_x$, and from $\tilde\lag_\mu = \tilde\lag_x^\Omega\cap \tilde\lag_x$. 
\end{proof}

\begin{proposition} \label{pro:A2enough} If, in the situation of Theorem~\ref{thm:polysymred2}, the condition $(A2)$ is satisfied, then so is also condition (\ref{eq:nondegenerate}). In that case, the reduced space $M_\mu=J^{-1}(\mu)/G_\mu$ admits the unique polysymplectic structure $\omega^a_\mu$, satisfying $\pi_\mu^*\omega^a_\mu=j_\mu^*\omega^a$.  
\end{proposition}
 
\begin{proof} We first fix $a=1$ and equip the manifold $M$ with the closed (but possibly degenerate) 2-form $\omega^1$. The pair $(M,\omega^1)$ is then a presymplectic manifold and it is clear from the definitions that $\Phi_g$ is a presymplectic action with momentum map $J^1\colon M\to\lag^*$ (the first component of $J$). We may therefore apply Lemma~\ref{lem:presymplectickernel} which implies that:
\[
\ker \left.j_{\mu_1}^*\omega^1\right|_x=\ker \left.{\omega^1}\right|_x+\left.\tilde\lag_{\mu_1}\right|_x. 
\]
Here, we write $j_{\mu_1}$  for the inclusion $(J^1)^{-1}(\mu_1)\to M$.  Note that $\mu_1$ is regular for $J^1$ because $\mu$ is regular for $J$.  We can make the same reasoning for each $a=1,\dots,k$, to obtain similar identities for each $a$. 

For convenience, let us denote:
\[
(\star)\equiv \cap_{a=1}^k \big(\left.\tilde\lag_{\mu_a}\right|_x+\left.\ker\omega^a\right|_x\big)\cap T_x(J^{-1}(\mu)).  
\]
Combining all of the identities for $\ker \left.j_{\mu_a}^*\omega^a\right|_x$ we find:
\begin{align*}
(\star)&=\cap_{a=1}^k\big(
\ker \left.j_{\mu_a}^*\omega^a\right|_x\big)\cap T_x(J^{-1}(\mu))\\
&=\big\{v\in T_xJ^{-1}(\mu) \st \omega^1_x\big(v,T_x(J^1)^{-1}(\mu_1)\big)=\dots=\omega^k_x\big(v,T_x(J^k)^{-1}(\mu_k)\big)=0\big\}\\
&\supset \big\{v\in T_xJ^{-1}(\mu) \st \omega^1_x(v,T_xJ^{-1}(\mu))=\dots=\omega^k_x(v,T_xJ^{-1}(\mu))=0\big\}\\
&=\cap_{a=1}^k\big(\left.j_\mu^*\omega^a\right|_x\big)=\tilde\lag_x^\omega \cap \tilde\lag_x^{\omega\omega}.
\end{align*}
We have used that $(J^1)^{-1}(\mu_1)\supset J^{-1}(\mu)$ (and similar for each $a$). We conclude that, if $(A2)$ holds, we   have the inclusion
\[
\left.\tilde\lag_\mu\right|_x \supset \tilde\lag_x^\omega \cap \tilde\lag_x^{\omega\omega}. 
\]
We had already mentioned  that $\left.\tilde\lag_\mu\right|_x= \tilde\lag_x\cap\tilde\lag_x^{\omega}$. Since $\tilde\lag_x\subset\tilde\lag_x^{\omega\omega}$, we always have $
\left.\tilde\lag_\mu\right|_x \subset \tilde\lag_x^\omega \cap \tilde\lag_x^{\omega\omega}$. With that, the statement of the proposition easily follows.
\end{proof}

It is instructive to say a few words about the reasons why the condition $(A1)$ appears in the work \cite{Polyreduction}. The authors rely on the following Lemma.

\begin{lemma}[Lemma 3.7 in \cite{Polyreduction}] \label{basiclemmaJC} Let $\Pi_a: V \to V_a$ be $k$ epimorphisms of real vector spaces of real dimension.  Assume there exists a symplectic structure $\Omega_a$ on each of the $V_a$ and $\cap_{a=1}^k \ker \Pi_a = \{0\}$, then $(V, \omega^a = \Pi^*\Omega_a)$ is a polysymplectic vector space.   
\end{lemma}

Given a polysymplectic structure on $M$, they then show that there exist a symplectic structure on the  vector space 
\[
V_a = \frac{\left( \frac{\ker T_xJ^a}{\ker \omega^a(x)} \right)}{\left.{\tilde\lag}_{\mu_a}\right|_x}.
\]
The proof of this property relies essentially on the same kind of arguments as our Lemma~\ref{lem:presymplectic}. By chasing diagrams, these spaces can be related to the vector space $V=T_{\pi_\mu(x)}(J^{-1}(\mu)/G_{\mu})$ by means of a certain linear map $\Pi_a={\tilde\pi}^a_x: V \to V_a$. The line of thought in \cite{Polyreduction} is the following: if this linear map satisfies the  conditions of Lemma~\ref{basiclemmaJC}, then we have found a polysymplectic structure on the vector spaces $T_{\pi_\mu(x)}(J^{-1}(\mu)/G_{\mu})$ (for each $x$), and thus also on the manifold $M_\mu$. The condition $(A1)$ is the translation, to this context, of the condition that the map $\Pi_a$ is an epimorphism, and the condition $(A2)$ is equivalent with the condition $\cap_{a=1}^k \ker \Pi_a = \{0\}$.

Proposition~\ref{pro:A2enough}, however, sheds some new light on the matter. Since we have now shown, in a direct way, that $(A2)$ is enough to ensure that we have a polysymplectic structure on $M_\mu$, there is actually no need to invoke the Lemma~\ref{basiclemmaJC} to pullback and glue the symplectic forms into a polysymplectic form. The sufficient condition $(A1)$ can be removed from Theorem~\ref{thm:polysymred2}.

In fact, the next example shows that the situation described in Lemma~\ref{basiclemmaJC} may not always be representative for the situation at hand. We show that there exist symplectic structures on $V_a$ and a polysymplectic structure on $V$, pulled back from $V_a$ by means of linear maps $\Pi_a$, even though the maps $\Pi_a$ are not epimorphisms.

\begin{example} We consider  $V=V_1=\R^4$ with basis $\{e_i\}_{i=1,\ldots 4}$ and dual basis $\{e^i\}_{i=1,\ldots 4}$ and $V_2=\R^2$ with basis $\{f_i\}_{i=1, 2}$ and dual basis $\{f^i\}_{i=1, 2}$. We consider the polysymplectic structures 
\[
\omega^1=e^1\wedge e^2,\qquad  \omega^2=e^3\wedge e^4,
\]
on $V$ and the symplectic structures   $\Omega_1=e^1\wedge e^2+e^3\wedge e^4$ on $V_1$ and $\Omega_2=f^1\wedge f^2$ on $V_2$. For the maps $\Pi_1$ and $\Pi_2$, we take
\[
\Pi_1(v^1,v^2,v^3,v^4)=(v^1,v^2,0,0)\Rightarrow \Pi_1^*\Omega_1=e^1\wedge e^2=\omega^1, 
\]
and
\[
\Pi_2(v^1,v^2,v^3,v^4)=(v^3,v^4)\Rightarrow \Pi_2^*\Omega_2=e^3\wedge e^4=\omega^2. 
\]  
$\Pi_2$ is surjective, but $\Pi_1$ is not. Also, note that $\ker \Pi_1\cap \ker\Pi_2=\{(0,0,0,0)\}$.
\end{example}

So far, we have only discussed cosymplectic reduction at the level of the geometric structures. There is also a part of the cosymplectic reduction theorem that deals with  Hamilton's equations.

A $k$-vector field $\pmb{X}=(X_1,\dots,X_k)$ is $G$-invariant if each component of $\pmb{X}$ is $G$-invariant as a vector field, or
\[
(T\Phi_g)X_a=X_a,\qquad a=1,\dots,k. 
\]
When $X_a$ is tangent to $J^{-1}(\mu)$ for each $a$, then the restriction of $\pmb{X}$ to $J^{-1}(\mu)$ defines a $k$-vector field on $J^{-1}(\mu)$. We will use the notation $\pmb{X}_{\mu}$ for this vector field. Its components are given by $(X_\mu)_a=\left.X_a\right|_{J^{-1}(\mu)}$. If $H\colon M\to \R$ is a $G$-invariant Hamiltonian, the reduced Hamiltonian $h_\mu\colon M_\mu\to\R$ is uniquely defined by the relation $\pi_\mu^*h_\mu=H_\mu$, with $H_\mu=j_\mu^*H$ the restriction of the Hamiltonian to the level set of $\mu$. As far as the dynamics is concerned, the following important consequence of Theorem~\ref{thm:polysymred} can be found in~\cite{Polyreduction}:

\begin{theorem}\label{thm:polysymdyn} Let $(M,\omega^a,\Phi_g,J)$ be a polysymplectic Hamiltonian $G$-space, and $\mu\in\lagdk$ a regular value of $J$. Assume that $G_\mu$ acts freely and properly on $J^{-1}(\mu)$ and that $(A1)$ and $(A2)$ hold. Let $H\colon M\to\R$ be a $G$-invariant Hamiltonian. Let $\pmb{X}$ be a solution of~\eqref{eq:k-Sym} for the Hamiltonian $H$ which is tangent to $J^{-1}(\mu)$ and $G_\mu$-invariant. Then, the projection $\overline{\pmb{X}}_{\mu}$ of $\pmb{X}_\mu$ on $M_\mu$ is a solution of~\eqref{eq:k-Sym} for the reduced Hamiltonian $h_\mu$. 
\end{theorem}

 Even though  Theorem~\ref{thm:polysymdyn} only applies to a very restricted class of solutions of \eqref{eq:k-Sym} (those that are tangent to $J^{-1}(\mu)$ and $G_\mu$-invariant), there do exist many  situations of interest in which Theorem~\ref{thm:polysymdyn} can be applied, see~\cite{Polyreduction,Polyrouth} for a few examples. Finally, it is now clear from Proposition~\ref{pro:A2enough}  that it is enough to verify $(A2)$ in Theorem~\ref{thm:polysymdyn}. 

\section{Reduction of polycosymplectic manifolds}\label{sec:reduction}

We now discuss how the reduction theory for polysymplectic manifolds extends to the case of polycosymplectic manifolds. To do that, we first mimic the approach of~\cite{Blackerpoly} to obtain a general reduction result. Later, we will adapt a technique found in \cite{singularcosymplectic} to find more practical sufficient conditions for the reduction, using the results of~\cite{Polyreduction}. 
 
\subsection{The linear case} 
 
Recall that a $k$-polycosymplectic vector space~\cite{bookpoly} is a vector space $V$ with a family of 2-forms $\omega^1,\dots,\omega^k$ on $V$, and a family of 1-forms $\eta^1,\dots,\eta^k$ on $V$ such that:
\begin{enumerate}[label=({\roman*})]
\item $\eta^1\wedge\dots\wedge\eta^k\neq 0$,
\item $\ker\omega^1\cap\dots\cap\ker \omega^k\cap\ker\eta^1\cap\dots\cap\ker \eta^k=\{0\}$,
\item $\dim\{\ker\omega^1\cap\dots\cap\ker \omega^k\}=k$.
\end{enumerate}
We let $R_1,\dots,R_k\in V$ be the Reeb vectors, uniquely defined by the following relations (for each $a,b$):
\begin{equation*}
i_{R_a}\omega^b=0,\qquad i_{R_a}\eta^b=\delta_a^b. 
\end{equation*}
We write $\disR={\rm span}\langle R_1,\dots,R_k\rangle=\{\ker\omega^1\cap\dots\cap\ker \omega^k\}\subset V$ for the subspace they generate. We will also use the following property: for each $a$,  $\eta^a$ is a 1-form whose kernel is generated by $R_a$. Therefore, we can write:
\begin{equation}\label{eq:etaidentity}
\ker\eta^a\oplus {\rm span}\langle R_a\rangle=V. 
\end{equation}

Similar to the polysymplectic case, if $S\subset V$ is a subspace, we define the \emph{polycosymplectic orthogonal of $S$} as the following subspace
\begin{equation*}
S^{c\omega}=\{v\in V\st \omega^1(v,S)=\dots=\omega^k(v,S)=0\}\subset V.   
\end{equation*}
The polycosymplectic orthogonal satisfies a number of properties similar to those in the polysymplectic case~\cite{Blackerpoly}. Among others, the following will be useful: for any subspace $S\subset V$, the inclusions $\disR\subset S^{c\omega}$ and $S\subset S^{c\omega\, c\omega}$ hold.

The next proposition is a generalization  of the essential part of
 Theorem 2.14 of  \cite{Blackerpoly}, when carried over  to the current setting.
\begin{proposition}\label{pro:linearreduction} Let $(V,\omega^a,\eta^a)$ be a $k$-polycosymplectic vector space and $S\subset V$ a subspace such that $\disR\cap S=\{0\}$.
Then the 2-forms $\omega^a$ and the 1-forms $\eta^a$ descend to 2-forms $\omega^a_S$ and 1-forms $\eta_S^a$ on the quotient $S^{c\omega}/(S\cap S^{c\omega})$. The tuple $(S^{c\omega}/(S\cap S^{c\omega}),\omega_S^a,\eta_S^a)$ is a $k$-polycosymplectic vector space if, and only if,
\[
\disR\oplus(S\cap S^{c\omega})= S^{c\omega}\cap S^{c\omega\,c\omega}. 
\]
\end{proposition}

\begin{proof} We first check that $\omega^a$ and $\eta^a$ descend to the quotient $S^{c\omega}/(S\cap S^{c\omega})$. If $v,w\in S^{c\omega}$ and $v',w'\in S\cap S^{c\omega}$ we have
\[
\omega^a(v+v',w+w')= \omega^a(v,w)+\omega^a(v,w')+\omega^a(v',w)+\omega^a(v',w')=\omega^a(v,w).
\]
On the other hand, in view of the equality~\eqref{eq:etaidentity} and the hypothesis $\disR\cap S=\{0\}$, for each $a$ we have $S\subset\ker\eta^a$. In particular $S\cap S^{c\omega}\subset\ker\eta^a$, and then $\eta^a(v+v')=\eta^a(v)$. 

We denote by $\pi\colon S^{c\omega}\to S^{c\omega}/(S\cap S^{c\omega})$ the quotient projection and by $j\colon S^{c\omega}\to V$ the canonical inclusion. We may then define  forms $\omega^a_S$ and $\eta_S^a$, by  the characterizing relations
\begin{equation}\label{eq:reducedlinear}
\pi^*\omega^a_S=j^*\omega^a,\qquad \pi^*\eta^a_S=j^*\eta^a. 
\end{equation}

We now show that $\dim(\ker\omega^1_S\cap\dots\cap\ker \omega^k_S)=k$. Consider $z\in S^{c\omega}$. If $\pi(z)\in\ker \omega_S^1\cap \dots\cap \ker \omega_S^k$, then from~\eqref{eq:reducedlinear} we have $z\in S^{c\omega}\cap S^{c\omega\, c\omega}$, and vice versa. From this, we get that 
\[
\ker\omega_S^1\cap \dots\cap \ker \omega_S^k=\pi( S^{c\omega}\cap S^{c\omega\,c\omega}).
\]
From the assumption $S^{c\omega}\cap S^{c\omega\,c\omega}=\disR\oplus(S\cap S^{c\omega})$, we get that $\pi(S^{c\omega}\cap S^{c\omega\,c\omega}) = \pi(\disR)$. We have already mentioned that $\disR\subset S^{c\omega}$. Moreover, given that $\disR\cap S=\{0\}$ (and in particular, $\disR\cap(S\cap S^{c\omega})=\{0\}$), we see that $\pi(R_1),\dots,\pi(R_k)$ are linearly independent. We conclude that $\dim(\pi(\disR))=k$, and therefore also  $\dim(\ker\omega_S^1\cap \dots\cap \ker \omega_S^k)=k$. 

Checking that 
\[
\eta^1_S\wedge \dots\wedge \eta^k_S\neq 0 
\]
and 
\[
\ker\omega_S^1\cap\dots\cap\ker \omega_S^k\cap\ker\eta_S^1\cap\dots\cap\ker \eta_S^k=\{0\} 
\]
is a  straightforward task.  \end{proof}

We end this section by remarking that, for $k=1$, the definition of $k$-polycosymplectic vector space reduces to that of a cosymplectic vector space. In particular, $\omega$ defines a presymplectic form in that case. Lemma~\ref{lem:presymplectic} then guarantees that any subset $S\subset V$ satisfies  
 \[
S^{c\omega\,c\omega}=S+\disR. 
 \] 

For a general polycosymplectic manifold (with general $k$) the above property may not hold. The following example demonstrates that there are manifolds with $\disR\oplus(S\cap S^{c\omega})= S^{c\omega}\cap S^{c\omega\,c\omega}$, but $S^{c\omega \, c\omega}\neq S+\disR$.

\begin{example} Consider  $V = \R^6= \R^3\times \R^3$. For $v=({\bf v},\vec v)$ and $w=({\bf w},\vec w)$, we  define $(\omega^a,\eta^a)$ ($a=1,2,3$) by the relations $\sum_{a=1}^3\omega^a(v,w) {\vec e}_a = \vec v\times \vec w \in \R^3$ and $\sum_{a=1}^3\eta^a(v){\bf e}_a = {\bf v}$. For this example is $\disR = \R^3 \times \{\vec 0\}$. Then, for the plane $S$ spanned by the vectors $({\bf 0},{\vec w}_1)$ and $({\bf 0},{\vec w}_2)$, we get that $S^{c\omega} = \{v\in \R^3\times\R^3 \st \vec v \times {\vec w}_1 = \vec v \times {\vec w}_2 = 0\} = \{({\bf v} ,\vec 0) \st {\bf v} \in\R^3\} =\R^3\times\{\vec 0\}$. But $S^{c\omega\,c\omega}=(\R^3\times\{0\})^{c\omega} = \R^3\times\R^3 \neq S \oplus \disR$, even though 
$S\cap\disR=\{({\bf 0},\vec 0)\}$ 
and
\[\disR\oplus (S\cap S^{c\omega}) = (\R^3 \times \{\vec 0\}) \oplus \{({\bf 0},\vec 0)\} = \R^3 \times \{\vec 0\} = S^{c\omega} \cap S^{c\omega\,c\omega}.
\] 
In this case $S^{c\omega}/(S\cap S^{c\omega}) =(\R^3\times\{\vec 0\})/\{({\bf 0},\vec 0)\} = \R^3\times\{\vec 0\}$. One easily verifies that  $\omega_S^a=0$ and $\sum_{a=1}^3\eta^a_S ({\bf v},\vec 0){\bf e}_a= {\bf v}$  define  a 3-polycosymplectic structure on $\R^3\times\{\vec 0\}$. 
\end{example}

\subsection{A polycosymplectic reduction theorem}\label{sec:polycosympletic}

The notion of $k$-cosymplectic manifold was introduced in~\cite{Hamiltoniansystemskcosymplectic} to overcome the limitations of polysymplectic manifolds and to be able to describe field theories in which the Hamiltonian or the Lagrangian can depend explicitely on the independent parameters. The definition of a $k$-cosymplectic manifold, much like that of $k$-symplectic manifold, includes an integrable distribution $\disR$ of a fixed dimension which imposes a dimensional constraint on the manifold (we will come back to this in Section~\ref{sec:kcosymplectic}). With regards to reduction, it is convenient to start from  an equivalent definition that does not include this distribution. Accordingly, we will use the notion of polycosymplectic manifold as it was introduced in~\cite{symmetrieskcosymplectic} (see also \cite{polycosymplecticmarreroreduction}). We will see that it is precisely within this category of polycosymplectic manifolds that reduction occurs.

\begin{definition}\label{def:polycosymplectic} A \emph{$k$-polycosymplectic structure}, or simply a \emph{polycosymplectic structure}, on a manifold $M$ is a family of closed two-forms $\omega^1,\dots,\omega^k$ and closed 1-forms $\eta^1,\dots,\eta^k$ such that: 
\begin{enumerate}[label=({\roman*})]
 \item $\eta^1\wedge\dots\wedge \eta^k\neq 0$,
 \item $\ker\omega^1\cap\dots\cap\ker \omega^k\cap\ker\eta^1\cap\dots\cap\ker \eta^k=\{0\}$,
 \item $\dim\{\ker\omega^1\cap\dots\cap\ker \omega^k\}=k$.
\end{enumerate}
We say that the polycosymplectic structure is standard if 
around  each point of $M$ there exists a coordinate chart $(t^a,q^i,p^a_i)$ such that
\begin{equation*}
\omega^a=dq^i\wedge dp_i^a,\qquad \eta^a=dt^a. 
\end{equation*}
\end{definition}
In particular, if $M$ is a $k$-polycosymplectic manifold then $\dim M=(k+1)n+k$ for some integer $n$. The local model for a standard polycosymplectic structure is the stable cotangent bundle of $k^1$-covelocities
\begin{equation*}
\R^k\times (T^1_k)^*Q ,
\end{equation*}
with the polycosymplectic structure $\omega^a=dq^i\wedge dp_i^a$ and $\eta^a=dt^a$, where $(q^i,p^a_i)$ are the usual coordinates on $(T^1_k)^*Q$ and $t^a$ the coordinates on $\R^k$.

For a polycosymplectic manifold, we may define the co-called Reeb vector fields $R_a$ as those that are uniquely characterized by the following relations (for each $a,b$):
\begin{equation}\label{eq:Reeb}
i_{R_a}\omega^b=0,\qquad i_{R_a}\eta^b=\delta_a^b.  
\end{equation}
We will denote by
\[
\disR= {\rm span}\left\{R_1,\dots,R_k\right\}=\ker\omega^1\cap\dots\cap\ker \omega^k.
\]
the distribution generated by the Reeb vector fields. It is an integrable distribution of rank $k$. 

\begin{example}\label{ex:cosymplectic} A cosymplectic manifold $(M,\Omega,\lambda)$ is a manifold $M$ of dimension $2n+1$ with a closed 2-form $\Omega$ and a closed 1-form $\lambda$ such that $\lambda\wedge (\Omega)^{n}$ is a volume form. Therefore, a cosymplectic manifold is a 1-polycosymplectic manifold. Since any cosymplectic manifold admits Darboux coordinates, they are always standard (see the appendix~B of \cite{bookpoly} for the statement).
\end{example}

\begin{example}\label{ex:product} Consider a family $(M_a,\Omega^a,\lambda^a)$, $a=1,\dots,k$, of cosymplectic manifolds. Denote by $T_a$  the Reeb vector field of $M_a$. We will now check that 
\[
M=M_1\times\dots\times M_k 
\]
is a $k$-polycosymplectic manifold, when endowed with the following family of 2-forms and 1-forms
\[
\omega^a={\rm pr}_a^*\Omega^a,\qquad \eta^a={\rm pr}_a^*\lambda^a,
\]
where ${\rm pr}_a\colon M\to M_a$ is the $a$-th projection. We identify 
\[
TM={\rm pr}_1^*(TM_1)\oplus\dots\oplus {\rm pr}_k^*(TM_k).
\]
Then
\[
\ker\omega^1\cap\dots\cap\ker \omega^k \simeq \ker\Omega^1\oplus \dots\oplus \ker\Omega^k.
\]
It has dimension $k$ as required, since  $\dim\ker\Omega^a=1$. Also, if
\[
v=v_1\oplus\dots\oplus v_k\in \big(\cap_{a=1}^k\ker\omega^a\big)\cap \big(\cap_{a=1}^k\ker\eta^a \big)
\]
then, for each $a$, we have $v_a\in\ker\Omega^a\cap\ker\lambda^a=\{0\}$. Finally, it is clear that $\eta^1\wedge\dots\wedge \eta^k\neq 0$. The Reeb vector $R_a$ is
\[
R_a=0\oplus\dots\oplus 0 \oplus T_a\oplus 0\oplus\dots\oplus 0, 
\]
with $T_a$ in the $a$-th position. Note that $M$ is, in general, not standard (although each $M_a$ is). For example, if we consider $M_1=M_2=\R^3$ with its canonical cosymplectic structure, then the product $M_1\times M_2$ is a 2-polycosymplectic manifold of dimension 6. It can not be standard because it is not of the form $3n+2$.
\end{example}

The dynamics in the polycosymplectic framework are obtained as follows. Given a Hamiltonian $H\colon M\to \R$, we say that a $k$-vector field $\pmb{X}$ solves the $k$-cosymplectic Hamilton's equations for $H$ if
\begin{equation}\label{eq:k-Cosym}
\eta^a(X_b)=\delta^a_b,\qquad \sum_{a=1}^k i_{X_a}\omega^a=dH-\sum_{a=1}^k R_a(H)\,\eta^a,\tag{k-Cosym}
\end{equation}
for each $a,b$. In the case where  $(M,\omega^a,\eta^a)$ is standard, one can show the following: for an integrable $k$-vector field which solves~\eqref{eq:k-Cosym}, the integral sections 
\[
t^a\mapsto (t^a,\psi^i,\psi^a_i) \in M 
\]
of $\pmb{X}$ satisfy the Hamilton-De~Donder-Weyl equations \eqref{eq:HdDW}. Again, we refer the reader to the monograph~\cite{bookpoly} for a complete discussion. It has been shown in \cite{RomanRey} that the equations \eqref{eq:k-Cosym} are consistent with \eqref{eq:k-Sym} in the following sense: for autonomous Hamiltonians (i.e. Hamitonians not depending explicitely on the independent parameters) these two sets of equations contain essentially the same solutions. 

\begin{remark}\label{rem:terminology} For $k=1$ (cosymplectic geometry), the unique solution of \eqref{eq:k-Cosym} is usually called the \emph{evolution vector field} and denoted $E_H$, and the term \emph{Hamiltonian vector field $X_H$} refers instead to the vector field such that $E_H=X_H+R$ (with $R$ the Reeb vector field), see for example \cite{CantrijnGradient,singularcosymplectic}. However, for general $k$, it is customary to call a solution of~\eqref{eq:k-Cosym} a \emph{Hamiltonian $k$-vector field}.
\end{remark}

\begin{definition}\label{def:polycosymplecticaction} An action $\Phi_g$ of a Lie group $G$ on a polycosymplectic manifold $(M,\omega^a,\eta^a)$ is \emph{polycosymplectic} if
\[
\Phi_g^*\omega^a=\omega^a,\qquad  i_{\xi_M}\eta^a=0,
\]
for each $a$ and $\xi\in\lag$.
\end{definition}
Definition \ref{def:polycosymplecticaction} extends that of \cite{Albert}, and agrees with the one in \cite{polycosymplecticmarreroreduction}, when $G$ is connected. Note that the second condition in Definition~\ref{def:polycosymplecticaction} is stronger than (and actually implies) the infinitesimal invariance condition $\mathcal{L}_{\xi_Q}\eta^a=0$. Recall that we are assuming that $\Phi_g$ is free and proper.  We also point out that, because of our assumption that Lie groups are connected, the condition $\mathcal{L}_{\xi_Q}\eta^a=0$ is equivalent to $\Phi_g^*\eta^a=\eta^a$.

An important observation is that, for a polycosymplectic action, the Reeb vector fields are invariant. This can be seen as follows: Using the defining relations~\eqref{eq:Reeb} for the Reeb vector fields an because $\omega^a$ and $\eta^a$ are invariant for a polycosymplectic action, we have
\[
i_{(\Phi_g)^*R_a}\omega^b=\Phi_g^*(i_{R_a}\omega^b)=0,\qquad i_{(\Phi_g)^*R_a}\eta^b=\Phi_g^*(i_{R_a}\eta^b)=\Phi_g^*(\delta_a^b)=\delta_a^b,
\]
and therefore $(\Phi_g)^*R_a=R_a$.

\begin{definition}\label{def:momentummap} A \emph{momentum map} for a polycosymplectic action $\Phi_g\colon M\to \lagdk$ is a map 
\[
J=(J^1,\dots,J^k)\colon M\to\lagdk 
\]
such that each $a$ and each $\xi\in\lag$ the following holds:
\begin{equation*} 
i_{\xi_M}\omega^a=dJ^a_\xi, 
\end{equation*}
where $J^a_\xi\colon M\to\R$ is the function $J^a_\xi(x)=\langle J^a(x),\xi\rangle$.
\end{definition}

Again, we need to say a few words about slightly different notions in the literature. It follows from Definition~\ref{def:momentummap} that $R_b(J^a_\xi)=0$ for each $a,b$ and $\xi\in\lag$. For $k=1$, this definition agrees with the usual notion of momentum map in cosymplectic geometry~\cite{Albert}: There, a momentum map  is a map $J\colon M\to\lag^*$ which satisfies $R(J_\xi)=0$ and which is such that each infinitesimal generator becomes a Hamiltonian vector field of the cosymplectic structure (in the sense of cosymplectic geometry, see Remark~\ref{rem:terminology}). We will come back later to Definition~\ref{def:momentummap} and show that it is a natural generalization to the polycosymplectic setting of the usual notion of momentum map in the sense that it leads to a conservation law. For the moment, we point out that if for each choice $\xi_1,\dots,\xi_k\in\lag$ we can define the function $\hat J\colon M\to\R$ as
\begin{equation*}
\hat J(\xi_1,\dots,\xi_n)(x)=J_{\xi_1}^1(x)+\dots+J_{\xi_k}^k(x). 
\end{equation*}
Then, we get
\[
i_{(\xi_1)_M}\omega^1+\dots+ i_{(\xi_k)_M}\omega^k=d\hat J(\xi_1,\dots,\xi_n),
\]
which is the defining relation of the momentum map in the paper \cite{polycosymplecticmarreroreduction}. The relation above means that the $k$-vector field $((\xi_1)_M,\dots,(\xi_k)_M)$ solves~\eqref{eq:k-Cosym} for the Hamiltonian $\hat J$. 

We will assume that the momentum map is equivariant w.r.t.\ the ${\rm Coad}^k_g$-action, i.e.\ that it satisfies~\eqref{eq:equivariance} (in the interpretation that is given to  $J$ and $\Phi_g$ in this section). In this case, since $G_\mu$ is the maximal subgroup of $G$ which stabilizes $J^{-1}(\mu)$, one has the equality:
\begin{equation*}
\left.\tilde\lag_\mu\right|_x= \tilde\lag_x\cap T_xJ^{-1}(\mu),
\end{equation*}
where
\[
J^{-1}(\mu)=\{x\in M\st J^a(x)=\mu_a,\;\; a=1,\dots,k\}. 
\]
Note that we are using the same notations as in Section~\ref{sec:polysympletic}. With similar arguments as there, one can show that $T_x J^{-1}(\mu) =\tilde\lag_x^{c\omega}$.
Thus, we always have:
\begin{equation}\label{eq:intersection}
\left.\tilde\lag_\mu\right|_x= \tilde\lag_x\cap \tilde\lag_x^{c\omega}.
\end{equation}
\begin{theorem}\label{thm:main} Let $\Phi_g$ be a polycosymplectic action on $(M,\omega^a,\eta^a)$ with equivariant momentum map $J$ and $\mu\in\lagdk$ a regular value of $J$. 
Then $M_\mu=J^{-1}(\mu)/G_\mu$ is a polycosymplectic manifold with polycosymplectic structure $(\omega^a_\mu,\eta^a_\mu)$  uniquely determined by the   relations 
\begin{equation}\label{eq:reducedforms}
\pi_\mu^*\omega^a_\mu=j_\mu^*\omega^a,\qquad  \pi_\mu^*\eta^a_\mu=j_\mu^*\eta^a, 
\end{equation}
if, and only if, for each $x\in J^{-1}(\mu)$ the following condition holds:
\begin{equation}\label{eq:nondegeneratepolyco}
\disR_x\oplus \left.\tilde\lag_\mu\right|_x=\tilde\lag_x^{c\omega} \cap \tilde\lag_x^{c\omega \,c\omega}.
\end{equation}
\end{theorem}

\begin{proof} We assume first that the condition~\eqref{eq:nondegeneratepolyco} is satisfied. Note that $\disR_x\cap \tilde\lag_x=\{0\}$ because of the condition $i_{\xi_M}\eta^a=0$ in Definition~\ref{def:polycosymplecticaction}. Under the assumption~\eqref{eq:nondegeneratepolyco} we have the following equalities:
\[
\tilde\lag_x^{c\omega} \cap \tilde\lag_x^{c\omega\, c\omega}=\disR_x\oplus\left.\tilde\lag_\mu\right|_x=\disR_x\oplus(\tilde\lag_x\cap \tilde\lag_x^{c\omega}).
\]
This means that we can apply Proposition~\ref{pro:linearreduction} to the subspace $S=\tilde\lag_x$.  For each $[x]=\pi_\mu(x)$, we get a polycosymplectic structure on
\[
\tilde\lag_x^{c\omega}/(\tilde\lag_x \cap \tilde\lag_x^{c\omega})=\tilde\lag_x^{c\omega}/\left.\tilde\lag_\mu\right|_x\simeq T_{[x]}(J^{-1}(\mu)/G_\mu).
\]
In the last step we have used the natural identification 
\[
T_xJ^{-1}(\mu)/\left.\tilde\lag_\mu\right|_x=T_xJ^{-1}(\mu)/T_x(G_\mu\cdot x)\simeq T_{[x]}(J^{-1}(\mu)/G_\mu), 
\]
which is obtained using $T_x\pi_\mu\colon T_xJ^{-1}(\mu)\to T_{[x]}(J^{-1}(\mu)/G_\mu)$. The forms $\omega_\mu^a$, $\eta_\mu^a$ are (uniquely) characterized by~\eqref{eq:reducedforms}, and in particular $\omega_\mu^a$, $\eta_\mu^a$ are closed. Therefore, $J^{-1}(\mu)/G_\mu$ is a polycosymplectic manifold.

The coverse statement easily follows from reversing the steps in the above reasoning, and from Proposition~\ref{pro:linearreduction}.

We remark that one may always reduce the forms $\omega^a,\eta^a$ to $M_\mu$ (regardless of whether~\eqref{eq:nondegeneratepolyco} holds).  But, these forms will only define a polycosymplectic structure if, and only if, the condition~\eqref{eq:nondegeneratepolyco} is satisfied. \end{proof}

We can now re-derive Abert's famous theorem, as a consequence of the previous result.

\begin{corollary}[Albert's cosymplectic reduction~\cite{Albert}]\label{cor:albert} Let $(M,\omega,\eta)$ be a cosymplectic manifold with a free and proper cosymplectic action and momentum map $J$. If $\mu\in\lag^*$ is a regular value of $J$, $M_\mu=J^{-1}(\mu)/G_\mu$ is a cosymplectic manifold and the cosymplectic structure $(\omega_\mu,\eta_\mu)$ is uniquely determined from the relations:
\begin{equation*}
\pi_\mu^*\omega^a_\mu=j_\mu^*\omega^a,\qquad  \pi_\mu^*\eta^a_\mu=j_\mu^*\eta^a.
\end{equation*}
\end{corollary}

\begin{proof} Recall that a 1-polycosymplectic manifold is a cosymplectic manifold. We need to check that condition~\eqref{eq:nondegeneratepolyco} is always satisfied in that case. Since a cosymplectic manifold is also  presymplectic with $\ker\omega=\disR$, we can use Lemma~\ref{lem:presymplectic} to write:
\[
\tilde\lag_x^{c\omega} \cap \tilde\lag_x^{c\omega\, c\omega}= \tilde\lag_x^{c\omega} \cap(\tilde\lag_x\oplus\disR_x)= (\tilde\lag_x^{c\omega} \cap\tilde\lag_x)\oplus\disR_x = \left.\tilde\lag_\mu\right|_x \oplus\disR_x.
\]
In the last two steps, we have used that $\disR$ lies in the cosymplectic orthogonal of any subspace and the identity~\eqref{eq:intersection}. 
\end{proof}

Let $\pmb{X}$ be a solution of~\eqref{eq:k-Cosym} for a $G$-invariant Hamiltonian $H$. If $J$ is a momentum map, and $\xi\in\lag$ is arbitrary, we find:
\[
\sum_a i_{X_a} dJ^a_\xi=\sum_a i_{X_a}i_{\xi_M} \omega^a=-i_{\xi_M}\left(\sum_a i_{X_a} \omega^a\right)=-i_{\xi_M}\left(dH-\sum_a R_a(H)\eta^a\right).
\]
We now observe that $i_{\xi_M}dH=\mathcal{L}_{\xi_M}H=0$ (since $H$ is invariant) and that $i_{\xi_M}\eta^a=0$ for each $a$. Since $J_\xi$ is a function on $M$, we can write $i_{X_a} dJ_\xi=\mathcal{L}_{X_a}J_\xi$ and we arrive at
\[
\sum_a\mathcal{L}_{X_a}J^a_\xi=0.
\]
In other words, the map $(J^1_\xi,\dots,J^k_\xi)\colon M\to\R^k$ is a conserved quantity in the field-theoretical sense (see~\cite{symmetrieskcosymplectic} for a full discussion). This is, essentially, the contents of Noether's Theorem in the current situation.

We will now obtain the second main result of this section, which describes in which way the dynamics  given by a Hamiltonian in a polycosymplectic manifold $(M,\omega^a,\eta^a)$ can be reduced to the quotient $(M_\mu,\omega_\mu^a,\eta_\mu^a)$. The definition of momentum map implies that the distribution $\disR$ lies inside $TJ^{-1}(\mu)$ or, in other words, that each Reeb vector field is tangent to $J^{-1}(\mu)$ . We will write $R_a^\mu$ for the restriction of $R_a$ to $J^{-1}(\mu)$, which is a well defined vector field on $J^{-1}(\mu)$.  We have already checked that the Reeb vector fields are $G$-invariant. From this, it also follows that their restrictions to $J^{-1}(\mu)$ are $G_\mu$-invariant.

In view of~\eqref{eq:reducedforms}, it is clear that the reduced Reeb vector fields $r_1,\dots,r_n$ on $M_\mu$, which satisfy 
\begin{equation*}
i_{r_a}\omega_\mu^b=0,\qquad i_{r_a}\eta_\mu^b=\delta_a^b, 
\end{equation*}
for each $a,b$, are given by $r_1=\pi_\mu(R^\mu_1),\dots,r_k=\pi_\mu(R^\mu_k)$. We will use the same notations as in  Section \ref{sec:polycosympletic}; in particular, we write $\pmb{X}_\mu$ for the restriction (when it exists) of $\pmb{X}$ to $J^{-1}(\mu)$.

\begin{theorem}\label{thm:maindynamics} Under the same conditions of Theorem~\ref{thm:main}, let $H\colon M\to\R$ be a $G$-invariant Hamiltonian and denote by $h_\mu$ its reduction to $M_\mu$. Let $\pmb{X}$ be a solution of~\eqref{eq:k-Cosym} for the Hamiltonian $H$ which is tangent to $J^{-1}(\mu)$ and $G_\mu$-invariant. Then the projection $\overline{\pmb{X}}_{\mu}$ of $\pmb{X}_\mu$ on $M_\mu$ is a solution of~\eqref{eq:k-Cosym} with Hamiltonian $h_\mu$.
\end{theorem}

\begin{proof} First note that $\overline{\pmb{X}}_{\mu}$ is a well-defined $k$-vector field on $M_\mu$ because of the assumed invariance of $\pmb{X}_\mu$. We write $(X_\mu)_a$ for the components of  $\pmb{X}_\mu$  and $(\overline{X}_\mu)_a= T\pi_\mu((X_\mu)_a)$ for the components of $\overline{\pmb{X}}_{\mu}$. Then, if $H_\mu=j_\mu^*H$, we have:
\begin{align*}
\pi_\mu^*\left(\sum_{a=1}^k i_{(\overline{X}_\mu)_a}\omega^a_\mu\right)&= \sum_{a=1}^k i_{(X_\mu)_a}\pi_\mu^*\omega^a_\mu= \sum_{a=1}^k i_{(X_\mu)_a}j_\mu^*\omega^a=\sum_{a=1}^k j_\mu^*(i_{X_a}\omega^a)=j_\mu^*\left(\sum_{a=1}^ki_{X_a}\omega^a\right)\\
&=j_\mu^*\left(dH-\sum_{a=1}^k R_a(H)\eta^a\right)=dH_\mu-\sum_{a=1}^k R^\mu_a(H_\mu) j_\mu^*\eta^a\\
&=\pi_\mu^*\left(dh_\mu-\sum_{a=1}^k r_a(h_\mu)\eta^a_\mu\right).
\end{align*}
But $\pi_\mu$ is a submersion, and therefore
\[
\sum_{a=1}^k i_{(\overline{X}_\mu)_a}\omega^a_\mu=dh_\mu-\sum_{a=1}^k r_a(h_\mu)\eta^a_\mu,
\]
as claimed.
\end{proof}

\section{The relation with polysymplectic reduction}\label{sec:relation}

In this section we will see that if $(M,\omega^a,\eta^a)$ is a $k$-polycosymplectic manifold, it is possible to define a $k$-polysymplectic structure on $\tilde M=M\times\R$ in such a way that the polycosymplectic reduction of $M$ is equivalent to the polysymplectic reduction of $\tilde M$. This fact will be used to derive some sufficient conditions for polycosymplectic reduction which are convenient when dealing with concrete examples.

\subsection{A related polysymplectic structure}

The key result of this section is the following lemma, which is an extension of Lemma 3.2 in \cite{singularcosymplectic} to general $k$:

\begin{lemma}\label{lem:extended} Let $M$ be a manifold, $\omega^1,\dots,\omega^k$ a family of closed 2-forms on $M$ and $\eta^1,\dots,\eta^k$ a family of closed 1-forms on $M$ which satisfy the following conditions:
\begin{enumerate}[label=({\roman*})]
\item $\eta^1\wedge\dots\wedge \eta^k\neq 0$,
\item $\dim\{\ker\omega^1\cap\dots\cap\ker \omega^k\}=k$.
\end{enumerate}
Consider the manifold $\tilde M=M\times\R$ and the family of 2-forms on $\tilde M$ defined by
\begin{equation}\label{def:tildeomega}
\tilde\omega^a={\rm pr}^*\omega^a+ds\wedge {\rm pr}^*\eta^a, 
\end{equation}
for each $a=1,\dots,k$, where $s$ is the coordinate in $\R$ and ${\rm pr}\colon \tilde M\to M$ is the canonical projection. 

Then $(\tilde M,\tilde\omega^a)$ is a polysymplectic manifold if, and only if, $(M,\omega^a,\eta^a)$ is polycosymplectic. 
\end{lemma}

\begin{proof} We make the usual identification $T\tilde M={\rm pr}^*(TM)\oplus {\rm pr'}^*(T\R)$, where ${\rm pr'}\colon \tilde M\to \R$ is the canonical projection. A tangent vector in $\tilde M$ at $\tilde x=(x,s)$ is of the form $(v_x,v_s)$, and 
\begin{equation}\label{eq:contraction}
i_{(v_x,v_s)}\tilde\omega^a={\rm pr}^*(i_{v_x}\omega^a)+ds(v_s)\,{\rm pr}^*\eta^a-{\rm pr}^*(i_{v_x}\eta^a)\, ds.  
\end{equation}
The 1-forms $i_{v_x}\omega^a$ and $\eta^a$ are independent: this can be seen from the identities 
\[
0=i_{R_a}(i_{v_x}\omega^a),\qquad 1=i_{R_a}\eta^a.
\]
Then, from~\eqref{eq:contraction} we can write 
\[
\left.\ker\tilde\omega^a\right|_{\tilde x} = \left.(\ker\omega^a\cap\ker\eta^a)\right|_{x}\oplus 0\subset T_{\tilde x}\tilde M,
\]
and it follows that $\cap_a\ker\tilde\omega^a=(0,0)$ if, and only if, $\cap_a(\ker\omega^a\cap\ker\eta^a)=0$.
\end{proof}

It is insightful to discuss the case where $(M,\omega^a,\eta^a)$ is standard in coordinates. Let us choose an adapted chart $U\subset M$ so that $M$ looks locally like $\R^k\times (T^1_k)^*Q$:
\[
\eta^a=dt^a,\qquad \omega^a=dq^i\wedge dp_i^a. 
\]
Then, on the chart $U\times\R\subset\tilde M$ we have
\[
\tilde\omega^a=dq^i\wedge dp_i^a+ds\wedge dt^a, 
\]
and it follows that $\tilde M$ is also standard. From this perspective, the 1-forms $\eta^a$ play the role of the momenta associated to the new coordinate $s\in\R$, giving the local identification with the canonical model $(T^1_k)^*(Q\times \R)$, where now the ``space of parameters'' $\tilde Q=Q\times\R$ has dimension $\dim Q+1$.  We will revisit this example later and show that this is actually a global statement in the sense that the lift of the stable cotangent bundle of $k^1$-covelocities of a manifold $Q$ is the cotangent bundle of $k^1$-covelocities of $Q$, see Section \ref{sec:stablecotangentbundle}.

An immediate observation about the way $\tilde\omega^a$ is defined in expression~\eqref{def:tildeomega} is that one can ``recover'' both $\eta^a$ and $\omega^a$ from $\tilde\omega^a$ as follows:
\begin{equation}\label{eq:recover}
\eta^a= (i_0)^*\big[i_{\partial/\partial s}\,\tilde\omega^a\big],\qquad \omega^a=(i_0)^*\tilde\omega^a, 
\end{equation}
where $i_0\colon M\to \tilde M = M\times\R$ is the inclusion at $0$, namely $i_0(x)=(x,0)$. Note that the first of the identities~\eqref{eq:recover} is obtained from $({\rm pr})^*\eta^a= i_{\partial/\partial s}\, \tilde\omega^a$  pulling back by $i_0$ and observing that $({\rm pr}\circ i_0)={\rm Id}_M$.

\begin{lemma}Let $(M,\omega^a,\eta^a)$ be a polyscosymplectic manifold and $\Phi_g$   a polyscosymplectic action on $M$ with momentum map $J$. Define $\tilde\Phi_g\colon \tilde M\to\tilde M$ by:
\[
\tilde\Phi_g(x,s)=(\Phi_g(x),s).
\]
Then: 
\begin{enumerate}[label=({\arabic*})]
\item $\tilde\Phi_g$ is a polysymplectic action.
\item The map $\tilde J\colon\tilde M\to\lagdk$ defined as $J(x,s)=J(x)$ is a momentum map for $\tilde\Phi_g$.
\end{enumerate}
\end{lemma}
\begin{proof} The first part follows easily from ${\rm pr}\circ \tilde\Phi_g=\Phi_g\circ {\rm pr}$. For the second part, note that $\tilde J^a_\xi={\rm pr}^* J^a_\xi$ and $\xi_{\tilde M}=(\xi_M,0)$. In particular, ${\rm pr}_*\xi_{\tilde M}=\xi_M\circ {\rm pr}$ and then
\[
i_{\xi_{\tilde M}}\tilde\omega^a={\rm pr}^*(i_{\xi_M}\omega^a-\eta^a({\xi_M}) ds)={\rm pr}^*(i_{\xi_M}\omega^a)=d\tilde J^a_\xi,
\]
where we have used that $i_{\xi_M}\eta^a=0$.
\end{proof}

We remark that the definition of $\tilde J$ implies that there is no ambiguity in the notation
\[
\tilde J^{-1}(\mu)=J^{-1}(\mu)\times\R. 
\]

Consider a polycosymplectic manifold $(M,\omega^a,\eta^a)$ with a polycosymplectic action $\Phi_g$ and momentum map $J$. We can construct, as above, the polysymplectic manifold $(\tilde M,\tilde\omega^a)$ and the action and momentum map $\tilde J$ and $\tilde \Phi_g$, respectively. In this situation, if we fix a regular value $\mu$ of $J$:
\begin{enumerate}[label=({\arabic*})]
\item We will say that $M$ can be reduced at $\mu$ if the conclusions of Theorem~\ref{thm:main} hold for $M$ at the regular value $\mu$ of $J$, i.e. if $M_\mu$ is a polycosymplectic manifold with polycosymplectic structure defined by $\pi_\mu^*\omega^a_\mu=j_\mu^*\omega^a$ and $\pi_\mu^*\eta^a_\mu=j_\mu^*\eta^a$. This means that we assume that $\disR_x\oplus \left.\tilde\lag_\mu\right|_x=\tilde\lag_x^{c\omega} \cap \tilde\lag_x^{c\omega \,c\omega}$. 
\item We will say that $\tilde M$ can be reduced at $\mu$ if the conclusions of Theorem~\ref{thm:polysymred} hold for $\tilde M$ at the regular value $\mu$ of $\tilde J$, i.e. if $\tilde M_\mu$ is a polysymplectic manifold with polysymplectic structure defined by $\tilde\pi_\mu^*\tilde\omega^a_\mu=\tilde j_\mu^*\tilde\omega^a$. This means that we assume that $\left.\tilde\lag_\mu\right|_{(x,s)}=\tilde\lag_{(x,s)}^{\tilde\omega} \cap \tilde\lag_{(x,s)}^{\tilde\omega\,\tilde\omega}$.
\end{enumerate}
Note that $\mu$ is a regular value of $J$ if, and only if, it is a regular value for $\tilde J$.

\begin{lemma} \label{newlemmatom} Assume that a polycosymplectic structure $\omega^a$ and $\eta^a$ on $M$ and a polysymplectic structure $\tilde \omega^a$ on $\tilde M=M\times \R$ are related as in Lemma~\ref{lem:extended}. Then
\[
\left.{\tilde\lag}_\mu \right|_{(x,s)} = \left.{\tilde\lag}_\mu \right|_{x} \times \{0\}, \qquad  {\tilde\lag}^{\tilde\omega}_{(x,s)} =  {\tilde\lag}^{c\omega}_{x} \times \R, \qquad  {\tilde\lag}^{\tilde\omega\,\tilde\omega}_{(x,s)} = (\cap_{a=1}^k\ker\eta^a_x \cap   {\tilde\lag}^{c\omega\,c\omega}_x) \times\{0\}
\] 
\end{lemma}

\begin{proof} From the definition of the extended action we see immediately that ${\tilde\lag}_{(x,s)} = {\tilde\lag}_x \times\{0\}$, and likewise $\left.{\tilde\lag}_\mu \right|_{(x,s)} = \left.{\tilde\lag}_\mu \right|_{x} \times \{0\}$.

Given that $\eta^a\mid_{\left.{\tilde\lag}_\mu \right|_{x}} = 0$, we get that $ {\tilde\lag}^{\tilde\omega}_{(x,s)} = {\tilde\lag}^{c\omega}_{x} \times \R$.

For the last property, consider $(v_x,v_s) \in  {\tilde\lag}^{\tilde\omega\,\tilde\omega}_{(x,s)}$. Then, for all  $(w_x,w_s) \in {\tilde\lag}^{\tilde\omega}_{(x,s)}= {\tilde\lag}^{c\omega}_{x} \times \R$,
\[
0={\tilde\omega}^a\left((v_x,v_s),(w_x,w_s)\right) = \omega^a(v_x,w_x) + v_s \eta^a(w_x) - w_s\eta^a(v_x).
\]
When we plug in $(w_x,w_s) = (0,1)$, we get that $v_x \in \cap_{a=1}^k\ker\eta^a_x$. 
When we plug in $(w_x,w_s) = (R_a,0)$, we get that $v_s=0$. Since now only $0=\omega^a(v_x,w_x)$, for all $w_x\in {\tilde\lag}^{c\omega}_{x}$ remains, we get the desired relation. 
\end{proof}

\begin{proposition}\label{pro:equivalencepolysympolyco} Let $(M,\omega^a,\eta^a)$ be a polyscosymplectic manifold. Then $(M,\omega^a,\eta^a)$ can be reduced at $\mu$ if, and only if, $(\tilde M,\tilde \omega^a)$ can be reduced at $\mu$.
\end{proposition}

\begin{proof} Assume that $(M,\omega^a,\eta^a)$ can be reduced at $\mu$.  In view of the previous lemma, we get
\begin{eqnarray*}
\tilde\lag_{(x,s)}^{\tilde\omega} \cap \tilde\lag_{(x,s)}^{\tilde\omega\,\tilde\omega} &=& \left((\cap_{a=1}^k\ker\eta^a_x \cap {\tilde\lag}^{c\omega\,c\omega}_x) \times\{0\} \right) \cap \left({\tilde\lag}^{c\omega}_{x} \times \R\right)
\\&=&  \left( \cap_{a=1}^k\ker\eta^a_x \cap \left( {\tilde\lag}^{c\omega\,c\omega}_x \cap  {\tilde\lag}^{c\omega}_{x}   \right) \right) \times\{0\} \\&=&
\left( \cap_{a=1}^k\ker\eta^a_x \cap \left( \disR_x \oplus  \left.{\tilde\lag}_\mu \right|_{x}    \right) \right) \times\{0\}
\\&=&
   \left.{\tilde\lag}_\mu \right|_{x}   \times\{0\} =    \left.{\tilde\lag}_\mu \right|_{(x,s)}.   
\end{eqnarray*}
In the before-last step, we have used that $\left.{\tilde\lag}_\mu \right|_{x}\subset \cap_{a=1}^k\ker\eta^a_x $ (in view of $\eta^a\mid_{\left.{\tilde\lag}_\mu \right|_{x}} = 0$) and $\cap_{a=1}^k\ker\eta^a_x \cap  \disR_x =\{0\}$.

Conversely, if we assume that $(\tilde M,\tilde \omega^a)$ can be reduced at $\mu$, then we see that  
\begin{equation} \label{converselift}  \cap_{a=1}^k\ker\eta^a_x \cap \left( {\tilde\lag}^{c\omega\,c\omega}_x \cap  {\tilde\lag}^{c\omega}_{x}   \right)   =
  \left.{\tilde\lag}_\mu \right|_{x}.  
\end{equation}
Recall that we only need to show that ${\tilde\lag}^{c\omega\,c\omega}_x \cap  {\tilde\lag}^{c\omega}_{x}  \subset \disR_x \oplus  \left.{\tilde\lag}_\mu \right|_{x}$. Let $v_x\in {\tilde\lag}^{c\omega\,c\omega}_x \cap  {\tilde\lag}^{c\omega}_{x}$.  Since we know that  $(\cap_{a=1}^k\ker\eta^a_x)\oplus\disR_x=T_xM$, we only need to consider two subcases. If $v_x\in \disR_x \subset \disR_x \oplus  \left.{\tilde\lag}_\mu \right|_{x}$, we are done. If $v_x$ lies in $\cap_{a=1}^k\ker\eta^a_x$, then the relation \eqref{converselift}  tells us that $v_x\in \left.{\tilde\lag}_\mu \right|_{x} \subset \disR_x \oplus  \left.{\tilde\lag}_\mu \right|_{x}.$
\end{proof}

We can now use the sufficient conditions in the polysymplectic reduction theorem of~\cite{polycosymplecticmarreroreduction} to find sufficient conditions for Theorem~\ref{thm:main}.

\begin{theorem}\label{thm:sufficient} Let $\Phi_g$ be a polycosymplectic action on $(M,\omega^a,\eta^a)$ with momentum map $J$ and $\mu$ a regular value. Assume that $J$ satisfies, for each $x\in J^{-1}(\mu)$ and $a=1,\dots,k$, the following  condition:
\begin{equation} \label{C1}
\left.\tilde\lag_\mu\right|_x=\cap_{a=1}^k \big( \left.\tilde\lag_{\mu_a}\right|_x  +\left.(\ker\omega^a\cap\ker\eta^a)\right|_{x}\big)\cap T_x(J^{-1}(\mu)). 
\end{equation}
Then $J^{-1}(\mu)/G_\mu$ is a polycosymplectic manifold and the polycosymplectic structure $(\omega^a_\mu,\eta^a_\mu)$ is uniquely determined from the relations:
\[
\pi_\mu^*\omega^a_\mu=j_\mu^*\omega^a,\qquad  \pi_\mu^*\eta^a_\mu=j_\mu^*\eta^a.
\]
\end{theorem}

\begin{proof} In view of Proposition~\ref{pro:A2enough} we only need to check that under the current condition on $J$ the condition $(A2)$ in Theorem~\ref{thm:polysymred2} is satisfied for $\tilde J$.

Recall from~\eqref{eq:contraction} that we have the following equality:
\[
\left.\ker\tilde\omega^a\right|_{\tilde x} = \left.(\ker\omega^a\cap\ker\eta^a)\right|_{x}\times \{0\}.
\]
Likewise, from Lemma~\ref{newlemmatom}  one can easily obtain the following two identities: \[
T_{\tilde x}(\tilde J^{-1}(\mu))= T_x(J^{-1}(\mu))\times \R
\]
and
\begin{align*}
\left.\tilde\lag_\mu\right|_{\tilde x}
&= \left.\tilde\lag_\mu\right|_x\times \{0\}\\
&=\left[\cap_{a=1}^k \big( \left.\tilde\lag_{\mu_a}\right|_x +\left.(\ker\omega^a\cap\ker\eta^a)\right|_{x}\big)\cap T_x(J^{-1}(\mu))\right]\times \{0\}\\
&= \left(\left[\cap_{a=1}^k \big( \left.\tilde\lag_{\mu_a}\right|_x +\left.(\ker\omega^a\cap\ker\eta^a)\right|_{x}\big)\right]\times \{0\}\right) \cap \left(T_x(J^{-1}(\mu)) \times \R\right)\\
&= \cap_{a=1}^k \big( \left.\tilde\lag_{\mu_a}\right|_x  +\left.\ker\tilde \omega^a\right|_{\tilde x}\big)\cap T_{\tilde x}(\tilde J^{-1}(\mu)).
\end{align*}
Therefore, the conditions of Theorem~\ref{thm:polysymred} apply to the polysymplectic action $\tilde \Phi_g$ on $\tilde M$ with momentum $\tilde J$. We can now apply Proposition~\ref{pro:equivalencepolysympolyco}.
\end{proof}

The situation is summarized in the following diagram
\begin{equation}\label{eq:dia1}
\begin{tikzcd}[row sep=large, column sep = 13ex]
M\arrow[hookrightarrow]{r}{i_0} & M\times\R \arrow[bend right=30,dashed,swap]{l}{{\rm pr}}  \\
    J^{-1}(\mu)\arrow[hookrightarrow]{r}\arrow{d}[swap]{\pi_\mu}\arrow{u}{j_\mu}\arrow{r}{i^\mu_0} & J^{-1}(\mu)\times\R \arrow{d}{\tilde\pi_\mu}\arrow{u}[swap]{\tilde j_\mu}  \arrow[bend left=30,dashed]{l}{{\rm pr}_\mu}  \\
    M_\mu\arrow{r}[hookrightarrow,swap]{\tilde i_0} & M_\mu\times\R \arrow[bend left=30,dashed]{l}{{\rm Pr}} 
\end{tikzcd}  
\end{equation} 
which can be used to give an alternative proof of Proposition~\ref{pro:equivalencepolysympolyco} through diagram chasing.

\subsection{Dynamics on \texorpdfstring{$\tilde M$}{Lg}}

There is an interesting, alternative approach to Theorem~\ref{thm:maindynamics}. Consider on $\tilde M$ the extended Hamiltonian $\tilde H\colon\tilde M\to \R$ defined as
\[
\tilde H={\rm pr}^*H-ks, 
\]
or $\tilde H(x,s)=H(x)-ks$. Let $\pmb{X}$ be a solution of \eqref{eq:k-Cosym} for $M$, and consider the $k$-vector field $\pmb{\tilde X}$ on $\tilde M$ with components:
\[
\tilde X_a=X_a+R_a(H)\fpd{}{s}\in\mathfrak{X}(\tilde M).
\]
We claim that $\pmb{\tilde X}$ solves \eqref{eq:k-Sym} in $(\tilde M,\tilde \omega^a)$ for $\tilde H$. Indeed, since 
\[
i_{\tilde X_a} \tilde\omega^a={\rm pr}^*(i_{X_a}\omega^a)-ds+{\rm pr}^*(R_a(H)){\rm pr}^*\eta^a,
\]
we find
\begin{align*}
\sum_{a=1}^k i_{\tilde X_a} \tilde\omega^a&={\rm pr}^*\left(dH-\sum_{a=1}^k R_a(H)\eta^a\right)-kds+\sum_{a=1}^k {\rm pr}^*(R_a(H)\eta^a)\\
&={\rm pr}^*dH - kds=d\tilde H. 
\end{align*}

If $\pmb{X}$ is tangent to $J^{-1}(\mu)$ and $G_\mu$-invariant, then $\pmb{\tilde X}$ is tangent to $\tilde J^{-1}(\mu)$, and it is invariant w.r.t. the lifted action $\tilde\Phi_g$ of $G_\mu$. We let $\pmb{\tilde Z}$ denote the reduced $k$-vector field on $\tilde M_\mu$ associated with $\pmb{\tilde X}$, which has components $\tilde Z^a=T\tilde\pi_\mu(\tilde X_a)$ (see Diagram~\eqref{eq:dia1} for the notations). From Theorem~\ref{thm:polysymdyn} we know that $\pmb{\tilde Z}$ solves~\eqref{eq:k-Sym} on $\tilde M_\mu$ for the Hamiltonian $\tilde h_\mu$ (the reduction of $\tilde H$). It is clear that $\overline{\pmb{X}}_{\mu}$ is the $k$-vector field on $M_\mu$ with components $Z^a= T{\rm Pr}(\tilde Z^a)$, where ${\rm Pr}\colon \tilde M_\mu\to M_\mu$ is the projection. From here one can check that $\overline{\pmb{X}}_{\mu}$ solves~\eqref{eq:k-Cosym} in $M_\mu$.

The previous discussion suggest a close relation between solutions of~\eqref{eq:k-Cosym} in $M$ for the Hamiltonian $H$, and solutions of~\eqref{eq:k-Sym} in $\tilde M$ for the Hamiltonian $\tilde H$. One should, however, pay special attention when relating this observation to solutions of field theories which, as we mentioned before, are given by integral sections. In fact, in general there is no way to lift a solution of the field theory in $M$ (an integral section of $\pmb{X}$) to a solution of the field theory in $\tilde M$ (an integral section of $\pmb{\tilde X}$). This phenomenon does not occur in the case $k=1$, when the lifted vector field $\tilde X$ contains all the information about the dynamics in the cosymplectic manifold $M$, see~\cite{singularcosymplectic}.

\begin{example}\label{ex:section} Consider the following Hamiltonian on $M=\R^2\times (T^1_2)^*\R$:
\[
H(t^1,t^2,q,p^1,p^2)=q (t^1t^2)+\frac{(p_1)^2 +(p_2)^2}{2}.
\]
This Hamiltonian is constructed as a two-dimensional version of the Hamiltonian governing the electrostatic equations for a charge distribution of the form $-t^1t^2$ (up to a constant), see Section 13.1 in \cite{bookpoly}. In this case, $\tilde M=M\times\R$ and 
\[
\tilde H(t^1,t^2,q,p^1,p^2,s)=-q(t^1t^2)-2s+\frac{(p_1)^2 +(p_2)^2}{2}.
\]
An integral section $\varphi\colon \R^2\to M$ 
\[
\varphi(t^1,t^2)= \big(t^1,t^2,\psi(t^1,t^2),\psi^1(t^1,t^2),\psi^2(t^1,t^2)\big)
\]
of a $k$-vector field $\pmb{X}$ which solves~\eqref{eq:k-Cosym} for $H$ will 
satisfy the Hamilton De Donder-Weyl equations~\eqref{eq:HdDW}:
\[
t^1t^2= \left(\fpd{\psi^1}{t^1}+\fpd{\psi^2}{t^2}\right),\qquad \psi^1=\fpd{\psi}{t^1}, \qquad \psi^2=\fpd{\psi}{t^2}.
\]
It is easy to check that the following integral section is a solution to the previous system of PDEs:
\[
\varphi(t^1,t^2)= \left(t^1,t^2,\frac{(t^1)^3t^2}{6},\frac{(t^1)^2t^2}{2},\frac{(t^1)^3}{6})\right). 
\]
Assume now that $\varphi$ is the projection of some integral section $\tilde\varphi\colon \R^2\to \tilde M$ of a solution $\tilde X$ of~\eqref{eq:k-Sym} for $\tilde H$. It necessarily has the form:
\[
\tilde\varphi(t^1,t^2)=\left(\psi=\frac{(t^1)^3t^2}{6},\psi^1=\frac{(t^1)^2t^2}{2},\psi^2=\frac{(t^1)^3}{6},\tilde\psi(t^1,t^2),\tilde\psi^1=t^1,\tilde\psi^2=t^2\right). 
\]
Then, again  in view of the Hamilton De Donder-Weyl equations~\eqref{eq:HdDW} (note that now ``$s$'' is a field with corresponding momenta $t^1,t^2$), the function $\tilde\psi$ must satisfy
\[
\fpd{\tilde\psi}{t^1}=-\psi t^2=-\frac{(t^1)^3(t^2)^2}{6},\qquad  \fpd{\tilde\psi}{t^2}=-\psi t^1=-\frac{(t^1)^4t^2}{6}.
\]
The equality of the mixed partial derivatives of $\tilde\psi$ w.r.t.\ $t^1$ and $t^2$ leads to an inconsistency. Thus, no such $\tilde\varphi$ exists.
\end{example}

The previous example is no exception. Indeed, a simple computation shows that the integrability of $\pmb{X}$ does not imply that of $\pmb{\tilde X}$, and there will be, in general, no integral sections of $\pmb{\tilde X}$ projecting onto those of $\pmb{X}$. The brackets $[\tilde X_a,\tilde X_b]$ satisfy, if $\pmb{X}$ is integrable, the following:
\begin{align*}
 0=[\tilde X_a,\tilde X_b]&=[X_a,X_b]+[X_a,R_b(H)\partial_s]+[R_a(H)\partial_s,X_b]+[R_a(H)\partial_s,R_b(H)\partial_s]\\
&=[X_a,R_b(H)\partial_s]+[R_a(H)\partial_s,X_b]=\big(X_a(R_b(H))-X_b(R_a(H))\big)\partial_s.
\end{align*}
This lack of integrability is closely related to the problem of reconstruction, and has already been described in the more general case of a principal bundle $P\to P/G$. The obstruction for a given lift $\pmb{Y}$ (a $k$-vector field on $P$) of some integrable $\pmb{X}$ (a $k$-vector field on $P/G$) to be integrable can be characterized as the vanishing of a certain curvature. In the current formalism, details can be found in~\cite{LTM_LP,Polyrouth}.

\section{Special cases and examples}\label{sec:examples}

\subsection{\texorpdfstring{$k$}{lg}-symplectic and \texorpdfstring{$k$}{lg}-cosymplectic manifolds}\label{sec:kcosymplectic}

The most important classes of polysymplectic and polycosymplectic manifolds are, respectively, the $k$-symplectic and $k$-cosymplectic manifolds. In this section we will show that in the case where $M$ is a $k$-cosymplectic manifold, $\tilde M$ is a $k$-symplectic manifold in a natural way. We first recall the basic definitions.

\begin{definition}\label{def:ksympl} 
Let $(M,\omega^a)$ be a $k$-polysymplectic manifold of dimension $\dim M=(k+1)n$ and $\disV$ an integrable distribution on $M$ of rank $nk$. We say that $(M,\omega^a,\disV)$ is a \emph{$k$-symplectic manifold} if $\omega^a(\disV,\disV)=0$ for each $a$.
\end{definition}
It has been shown~\cite{Awane} that a $k$-symplectic manifold $(M,\omega^a,\disV)$ admits a Darboux-type theorem for both the 2-forms $\omega^a$ and the distribution $\disV$: around every point in $M$ we can find adapted coordinates as in \eqref{eq:darboux-standard} such that, additionally,  $\disV={\rm span}\left\{\partial/\partial p^a_i\right\}$.

A $1$-symplectic manifold is an example of a polarized symplectic manifold, namely a symplectic manifold with a Lagrangian foliation $\mathcal{F}$ which, in the present case, is determined from the $n$-dimensional distribution tangent to the leaves of $\mathcal{F}$. By a well-known argument (see e.g.~\cite{Vaisman1}), the symplectic form and Lagrangian foliation can be brought locally -and simultaneously- to the form $\omega=dq^i\wedge dp_i$ and $\mathcal{F}=\{q^i={\rm constant}\}$,  and this gives the adapted coordinates  when $k=1$. Thus, as pointed out in~\cite{Awane}, the definition of $1$-symplectic manifold generalizes the notion of polarized symplectic manifold. As a matter of fact, $k$-symplectic structures have been called \emph{polarized $k$-symplectic structures} in~\cite{Submanifoldsk}. A $k$-symplectic manifold is an obvious way a standard polysymplectic manifold. The converse, however, need not be true: a standard polysymplectic manifold may not admit a distribution $\disV$ fulfilling the requirements in Definition~\ref{def:ksympl} (the sphere $\mathbb{S}^2\subset\R^3$ is an example of a symplectic manifold which admits no polarization). 

\begin{remark} Polysymplectic structures were introduced by G{\"{u}}nther~\cite{Gunter}. The definition of $k$-symplectic structure is due to Awane~\cite{Awane}, although the same structure appeared independently in~\cite{palmostcotangent,regularpalmostcotangent} with the name \emph{$k$-almost cotangent structure}. In the particular case of the frame bundle, polysymplectic structures are also described in the works of Norris (see e.g.~\cite{Norris1}).
\end{remark}

\begin{definition}\label{def:kcosymplectic} 
Let $(M,\omega^a,\eta^a)$ be a $k$-polycosymplectic manifold of dimension $\dim M=k(n+1)+n$ and $\disV$ an integrable distribution on $M$ of rank $nk$. We say that $(M,\omega^a,\eta^a,\disV)$ is a \emph{$k$-cosymplectic manifold} if $\omega^a(\disV,\disV)=0$\, $\eta^a(\disV)=0$ and $[R_a,\disV]\subset \mathcal{V}$ (for each $a$).
\end{definition}
Usually, the condition $[R_a,\disV]\subset \mathcal{V}$ is not included in the definition of $k$-cosymplectic manifold (see e.g.~\cite{Hamiltoniansystemskcosymplectic,MerinoPhd}). The reason is the following proposition, which renders that condition superfluous unless $k=1$:
\begin{proposition}[Lemma 5.1.1 in~\cite{MerinoPhd}]\label{pro:involutive} Let $(M,\omega^a,\eta^a)$ be a $k$-polycosymplectic manifold of dimension $\dim M=k(n+1)+n$ and $\disV$ an integrable distribution on $M$ of rank $nk$ such that $\omega^a(\disV,\disV)=0$. If $k\geq 2$, then $[R_a,\disV]\subset \mathcal{V}$. In particular, the distribution $\disR\oplus\mathcal{V}$ is involutive.
\end{proposition}

It follows from Definition~\ref{def:kcosymplectic} that for a $k$-cosymplectic manifold the distribution $\disR\oplus\mathcal{V}$ is involutive. The stable cotangent bundle of $k^1$-covelocities is an example of a $k$-cosymplectic manifold, where $\disV$ is the vertical space w.r.t. the projection $\R^k\times(T^1_k)^*Q\to\R^k\times Q$. It has been shown in \cite{Hamiltoniansystemskcosymplectic} that a $k$-cosymplectic manifold admits Darboux-type coordinates such that: 
\begin{equation}\label{eq:darboux-kycosymplectic}
\eta^a=dt^a,\qquad \omega^a=dq^i\wedge dp_i^a,\qquad \disV={\rm span}\left\{\fpd{}{p^a_i}\right\}. 
\end{equation}
It is apparent from this local formulae that one needs to add the condition $[R_a,\disV]\subset \mathcal{V}$  in Definition~\ref{def:kcosymplectic} (at least for $k=1$) to ensure the existence of Darboux coordinates. This fact seems to have been overlooked in some of the literature, and it is often claimed that a $k$-cosymplectic manifold admits Darboux-type coordinates as in~\eqref{eq:darboux-kycosymplectic} without mentioning that the case $k=1$ requires a further compatibility condition. 

A simple example where the conclusion of Proposition~\ref{pro:involutive} fails in the case $k=1$ can be found in~\cite{stable}. In $\R^3$ we consider the following data: 
\[
\eta=dt,\qquad \omega=dx\wedge dp,\qquad \disV={\rm span}\left\{\fpd{}{p}+t\fpd{}{x}\right\}. 
\]
This is readily checked to define a $1$-cosymplectic structure which is standard, but $[R,\disV] \not\subset \disV$. In particular, it cannot admit Darboux-type coordinates because $\disR\oplus \disV$ is not involutive.

Given a distribution $\disW$ on $M$, we denote by $\tilde\disW$ the induced (lifted) distribution on $\tilde M$: 
\[
\tilde\disW=\disW\oplus 0\equiv \{v\in T\tilde M\st (T{\rm pr})(v)\in \mathcal{W},\; (T{\rm pr'})(v)=0\} \subset T\tilde M. 
\]
In the case where $\disW$ is integrable so is $\tilde\disW$: for if $[X,Y]=Z$ holds for elements in $\disW$, then 
\[
[X\oplus 0,Y\oplus 0]=Z\oplus 0\in\tilde\disW. 
\]
The following is an interesting consequence of Lemma~\ref{lem:extended}:
\begin{proposition}\label{pro:liftstoksympl} If $(M,\omega^a,\eta^a,\disV)$ is a $k$-cosymplectic manifold, then $(\tilde M,\tilde\omega^a,\tilde\disW)$ is a $k$-symplectic manifold, where $\disW=\disV\oplus\disR$.
\end{proposition}
\begin{proof} Let $(M,\omega^a,\eta^a,\disV)$ be a $k$-cosymplectic manifold with $\dim M=(k+1)n+k$. By construction $\dim \tilde M=(k+1)(n+1)$. Since $\disV\cap\disR=\emptyset$ the rank of $\tilde\disW$ (which is the same as the rank of $\disW$) is $k(n+1)$. Also $\tilde\disW$ is involutive in view of Definition~\ref{def:kcosymplectic}. Finally if $X,Y\in\disW$:
\[
\tilde\omega^a(X\oplus 0,Y\oplus 0) =\omega^a(X,Y)=0.
\]
\end{proof}

\subsection{Stable cotangent bundle}\label{sec:stablecotangentbundle}

Let $Q$ be a manifold and $\Phi_g\colon Q\to Q$ an action. Consider the polycosymplectic manifold $M=\R^k\times (T^1_k)^*Q$ with the natural lifted action. We have already seen that the polysymplectic manifold $\tilde M$ can be identified with $(T^1_k)^*\tilde Q$ (where $\tilde Q=Q\times\R$) with its canonical polysymplectic structure which we denote by $\tilde\omega^a$. We make this identification more precise now, and show that both spaces are the same as bundles over $\tilde Q$. First, we regard $\tilde M$ as a bundle over $\tilde Q$ via the projection 
\begin{align*}
\tilde M\equiv  (T^1_k)^*Q\times\R^k\times \R &\to Q\times \R \\
(\alpha_q^1,\dots,\alpha_q^k,t^1,\dots,t^k,s)&\mapsto (q,s),
\end{align*}
where $\alpha_q^i\in T_q^*Q$. On the other hand, we have
\[
T^*\tilde Q={\rm pr}_1^*(T^*Q)\oplus {\rm pr}_2^*(T^*\R)\simeq {\rm pr}_1^*(T^*Q)\times\R\times \R,
\]
where the projection is $(\alpha^1_q,t^1,s)\mapsto (q,s)$ (and we think of $t^1\in T_s^*\R$). Therefore 
\[
(T^1_k)^*\tilde Q={\rm pr}_1^*((T^1_k)^* Q)\times \R^k,
\]
which is just $\tilde M$ with a pull-back notation. If we denote by $(q^i,s)$ the coordinates in $\tilde Q$ and by $(p_i^a,t^a) $ their corresponding momenta in $(T^1_k)^*\tilde Q$, then the polysymplectic structure on $\tilde M$ is
\[
\tilde\omega^a=dq^i\wedge dp_i^a+ds\wedge dt^a. 
\]
If we write $\Psi_g\colon \tilde Q\to \tilde Q$ for the action $\Psi_g=(\Phi_g,{\rm Id}_{\R})$, then its cotangent lift to $(T^1_k)^*\tilde Q$ coincides with $\tilde\Phi_g$. 

The components momentum map $J^a_\xi\colon \R^k\times (T^1_k)^*Q\to\R$ are obtained from the usual momentum map for cotangent bundles:
\begin{equation*}
J^a_\xi(t^1,\dots,t^k,\alpha_q^1,\dots,\alpha_q^k)=\langle \alpha_q^a,\xi_Q(q)\rangle. 
\end{equation*}
As expected, the coordinates $(t^1,\dots,t^k)$ play no role in the reduction, and one can easily relate the reduction of $M$ and that of the $(T^1_k)^*Q$, which has been studied in~\cite{Polyrouth,MunSal}. Consider for simplicity $\mu=(0,\dots,0)$. Then the reduced polysymplectic space $M_\mu$ can be identified with 
\[
M_\mu\simeq (T^1_k)^*(Q/G)\times \R^{k},
\]
equipped with its canonical polysymplectic structure. One may as well identify $M_\mu$ and compute its polysymplectic structure for more general values of $\mu$, but the identification requires, for example, the choice of a principal connection on the bundle $Q\to Q/G$. Alternatively, one may directly study the reduction of $\tilde M\simeq (T^1_k)^*\tilde Q$ and find (at $\mu=(0,\dots,0)$)
\[
\tilde M_\mu=(T^1_k)^*(\tilde Q/G)= (T^1_k)^*\big((Q/G)\times \R\big)\simeq (T^1_k)^*(Q/G)\times\R^{k+1}\simeq M_\mu\times\R,
\]
in agreement with the result above. We forego the proof that both manifolds have the same polycosymplectic structure, but this can be readily checked.

\subsection{Product of cosymplectic manifolds}\label{sec:product}

With the same notations as in Example~\ref{ex:product}, let us assume that each $M_a$ comes with a free and proper cosymplectic action $\Phi_a$ of a Lie group $G_a$ with equivariant momentum map $J^a$. Then, if $G=G_1\times \dots\times G_k$, we have the following $G$-action on $M$:
\[
\Phi_{(g_1,\dots,g_k)}(x_1,\dots,x_k)=\big(\Phi_{g_1}(x_1),\dots,\Phi_{g_k}(x_k)\big), 
\]
where $x=(x_1,\dots,x_k)\in M$. This action can be checked to be polycosymplectic. Let
\[
\lag=\lag_1\times\dots\times\lag_k  
\]
be the Lie algebra of $G$, where $\lag_a$ is the Lie algebra of $G_a$. Then $\Phi$  admits a momentum map $\mathbb{J}\colon M\to\lagdk$ with components:
\begin{align*}
\mathbb{J}^a\colon M&\to \lag^*=\lag_1^*\times\dots\times\lag_k^*,\\ 
x=(x_1,\dots,x_k)&\mapsto (0,\dots,0,J^a(x_a),0,\dots,0),
\end{align*}
i.e. $\langle \mathbb{J}^a(x),(\xi_1,\dots,\xi_k)\rangle=\langle J^a(x_a),\xi_a\rangle$. We fix a regular value $\nu_a\in\lag^*$ of $J^a$ for each $a$, and write $\mu_a=(0,\dots,\nu_a,\dots,0)\in\lag^*$ (with $\mu_a$ in the $a$-th position) and $\mu=(\mu_1,\dots,\mu_k)\in\lagdk$. Finally, let us also denote by $(\overline{M}_a,\overline{\Omega}^a,\overline{\lambda}^a)$ the cosymplectic reduction of $M_a$ at $\nu_a$. As above, we will write 
\[
\overline{M}=\overline{M}_1\times\dots\times \overline{M}_k 
\]
for the polycosymplectic manifold obtained from the product of the reduced cosymplectic manifolds $\overline{M}_1,\dots,\overline{M}_k$. The polycosymplectic structure is given by 
\[
\overline{\omega}^a=\overline{{\rm pr}}_a^*\overline{\Omega}^a,\qquad \overline{\eta}^a=\overline{{\rm pr}}_a^*\lambda^a,
\]
where $\overline{{\rm pr}}_a\colon \overline{M}\to \overline{M}_a$ is the $a$-th projection.
\begin{proposition} With the notations above, there is an identification of polycosymplectic manifolds
\[
(M_\mu,\omega^a_\mu,\eta^a_\mu)\simeq (\overline{M},\overline{\omega}^a,\overline{\eta}^a).
\]
In other words, the polycosymplectic reduction of a product of cosymplectic spaces is the product of the corresponding cosymplectic reduced spaces.
\end{proposition}

\begin{proof} We have
\begin{equation*}
\mathbb{J}^{-1}(\mu)=(J^1)^{-1}(\mu_1)\times\dots\times (J^k)^{-1}(\mu_k), 
\end{equation*}
and, because $G_\mu=G_{\mu_1}\times\dots\times G_{\mu_k}$, there is an identification 
\[
M_\mu= \mathbb{J}^{-1}(\mu)/G_\mu\simeq\overline{M}_1\times\dots\times \overline{M}_k.
\]
To analyze the reduced polycosymplectic structure, we will use the notations in the following diagram:
\begin{equation*}
\begin{tikzcd}[row sep=large, column sep = 13ex]
&   & (J^1)^{-1}(\mu_1)\arrow[hookrightarrow]{r}{j_1}\arrow[swap,bend right=15]{lld}{\pi_1} &  M_1 \\
 \overline{M}_1 & \overline{M}_1\times\dots\times \overline{M}_k\arrow{l}{\overline{{\rm pr}}_1} & \mathbb{J}^{-1}(\mu)\arrow{l}{\pi_\mu}\arrow[hookrightarrow]{r}{j_\mu=(j_1,\dots,j_k)}\arrow{u}{\Pi_1} &  M\arrow[swap]{u}{{\rm pr}_1}
\end{tikzcd}  
\end{equation*} 
The reduced cosymplectic structure in $\overline{M}_1$ is characterized by the relations:
\[
\pi_1^*\overline{\Omega}^1=j_1^*\Omega^1,\qquad \pi_1^*\overline{\lambda}^1=j_1^*\lambda^1. 
\]
We will now check that $\omega^1_\mu=\overline{{\rm pr}}_1^*\overline{\Omega}^1$ and $\eta^1_\mu=\overline{{\rm pr}}_1^*\overline{\lambda}^1$. Indeed, we have
\[
j_\mu^*\omega^1=j_\mu^*{\rm pr}_1^*\Omega^1=\Pi_1^*j_1^*\Omega^1=\Pi_1^*\pi_1^*\overline{\Omega}^1=\pi_\mu^* \overline{{\rm pr}}_1^*\overline{\Omega}^1,
\]
and similar for $\eta^1_\mu$. The same reasoning applies to the rest of the forms $\omega_\mu^a,\eta_\mu^a$.
\end{proof}

It is possible to show that the sufficient condition of Theorem~\ref{thm:sufficient} is satisfied in this example (see~\cite{Poly-deLucas} where this is done in detail).

\paragraph{Acknowledgements.} We are indebted to the authors of \cite{Poly-deLucas}  for sharing the outline of their results before they came online. We are grateful to Javier de Lucas for pointing out a few inaccuracies in an earlier version of this manuscript. We thank the referees for their comments on the paper. E.\ García-Tora\~no Andr\'es is thankful to FONCYT for funding through
project PICT 2019-00196.

\bibliographystyle{plain}

\end{document}